\documentclass[a4paper]{article}

\title{On path-based coalgebras and weak notions of bisimulation}
\author{Harsh Beohar Sebastian K\"{u}pper\\
Theoretical Computer Science Group\\
Universit\"{a}t Duisburg-Essen\\
\texttt{\{harsh.beohar,sebastian.kuepper\}@uni-due.de}}

\usepackage[UKenglish]{babel}
\usepackage{csquotes,xpatch}
\usepackage[backend=bibtex,hyperref=true,bibencoding=ascii,bibstyle=numeric,sortcites,useprefix=true,doi=false,isbn=false,maxcitenames=2,maxbibnames=9]{biblatex}

\usepackage{stmaryrd,amssymb,amsmath,amssymb,amsthm,a4wide}
\usepackage{diagrams}
\usepackage{wrapfig}

\usepackage[utf8]{inputenc}
\usepackage[T1]{fontenc}

\usepackage{tikz}
\usetikzlibrary{shapes,arrows}
\usetikzlibrary{calc}

\newtheorem{theorem}{Theorem}
\newtheorem{lemma}[theorem]{Lemma}
\newtheorem{example}[theorem]{Example}
\theoremstyle{plain}
\newtheorem{proposition}[theorem]{Proposition}
\newtheorem{corollary}{Corollary}

\theoremstyle{definition}
\newtheorem{definition}{Definition}
\newtheorem{remark}{Remark}

\newcommand{\act}{A}
\newcommand{\tact}{A_\tau}
\newcommand{\trace}[1]{\mathrm{trace}(#1)}
\newcommand{\history}[1]{\downarrow\!{#1}}
\newcommand{\future}[1]{\uparrow\!{#1}}
\newcommand{\dom}[1]{\mathrm{dom}(#1)}
\newcommand{\cod}[1]{\mathrm{cod}(#1)}
\newcommand{\eseq}[1]{\varepsilon_{#1}}
\newcommand{\step}[1]{\xrightarrow {#1}}
\newcommand{\exec}[1]{\mathrm{Exec}(#1)}
\newcommand{\last}[1]{\mathrm{last}(#1)}
\newcommand{\mcal}[1]{\mathcal {#1}}
\newcommand{\paths}[1]{\mathrm{Path}(#1)}

\newcommand{\stutter}[1]{{#1}^{\dag}}
\newcommand{\id}[1]{\mathrm{id}_{#1}}
\newcommand{\qset}[1]{{#1}_{\sim}}
\newcommand{\spaths}[1]{\qset{\mathrm{Path}}(#1)}
\newcommand{\tpaths}{\mathrm{Path}}
\newcommand{\tspaths}{\mathrm{Path}_\sim}
\newcommand{\inv}[1]{{#1}^{-1}}

\newcommand{\meas}{\mathbf{Meas}}
\newcommand{\set}{\mathbf{Set}}
\newcommand{\G}{\mathcal G}
\newcommand{\eval}[1]{\epsilon_{#1}}
\newcommand{\nnreals}{\left[0,\infty\right]}
\newcommand{\closure}[1]{\mathrm{cl}(#1)}
\newcommand{\fsubseteq}{\subseteq_{\mathrm{f}}}

\newcommand{\lpaths}[1]{\mathrm{Path}^{\scriptscriptstyle\infty}(#1)}
\newcommand{\slpaths}[1]{\mathrm{Path}^{\scriptscriptstyle\infty}_{\sim}(#1)}

\newcommand{\scott}[1]{{#1}_{{\scriptscriptstyle\infty}}}

\makeatletter
\providecommand*{\twoheadrightarrowfill@}{%
  \arrowfill@\relbar\relbar\twoheadrightarrow
}
\providecommand*{\xtwoheadrightarrow}[2][]{%
  \ext@arrow 0395\twoheadrightarrowfill@{#1}{#2}%
}
\newcommand{\steps}[1]{\xtwoheadrightarrow{#1}{}}
\newcommand{\stepsX}[1]{\steps{#1}_{\!\!X}}
\makeatother

\newarrow{Multi}....{>>}

\begin{document}

\maketitle

\begin{abstract}
  It is well known that the theory of coalgebras provides an abstract definition of behavioural equivalence that coincides with strong bisimulation across a wide variety of state-based systems. Unfortunately, the theory in the presence of so-called silent actions is not yet fully developed. In this paper, we give a coalgebraic characterisation of branching (delay) bisimulation in the context of labelled transition systems (fully probabilistic systems). It is shown that recording executions (up to a notion of stuttering), rather than the set of successor states, from a state is sufficient to characterise the respected bisimulation relations in both cases.
\end{abstract}
\section{Introduction}\label{sec:intro}

Since its inception, coalgebra-based modelling of systems provides a simple and abstract definition of behavioural equivalence that coincides with the so-called strong bisimulation relations across a wide variety of dynamical systems (see \cite{Rutten:2000:coalgebra} for an introduction). Two states are said to be behaviourally equivalent if they are mapped to a common point by a coalgebra homomorphism. Unfortunately, the theory in the presence of so-called \emph{silent} actions is not yet well developed, albeit some general constructions (with varying level of generality) characterising Milner's weak bisimulation \cite{Milner:1980:CCS} are proposed in the literature (see, for instance, \cite{Goncharov:2014:wbisim,Sokolova:2005:actiontype,Brengos:2015:journal,Brengos:2014,Brengos13:order-enriched}
and the references therein).

Another refinement of strong bisimulation is \emph{branching bisimulation} proposed by \citeauthor{bbisim:1996} \cite{bbisim:1996}, which is the coarsest equivalence (in the \citeauthor{vanGlabbeek:1993:specII} spectrum \cite{vanGlabbeek:1993:specII}) preserving the branching structure of a state \cite{vanGlabbeek:2001:btime}. In this context, we are unaware of any prior work that captured branching bisimulation in the framework of coalgebras. Moreover, a natural notion of behavioural equivalence should preserve the branching structure of a state just like strong bisimulation does in the absence of silent action.

\citeauthor{Bonchi2015:kill-epsilons} \cite{Bonchi2015:kill-epsilons} have considered silent transitions coalgebraically by removing them all together by considering the labels as words rather than single letters. This approach is not useful when characterising branching bisimulation (or even weak bisimulation) because not all silent transitions can always be removed from the system without violating the transfer properties of branching (weak) bisimulation.
In \cite{Goncharov:2014:wbisim,Brengos:2015:journal,Brengos:2014,Brengos13:order-enriched}, weak bisimulation is captured in two phases: first, a given coalgebra is transformed into a coalgebra (possibly over a different base category) which captures the ``saturation'' effect of a silent action; second, it is shown that the notion of behavioural equivalence on this transformed coalgebra coincides with weak bisimulation on the corresponding dynamical system. In \cite{Sokolova:2005:actiontype} the authors used Aczel-Mendler style formulation of strong bisimulation in the latter step.

\begin{figure}
  \centering
  \vspace{-0.5cm}
  \includegraphics[width=5cm]{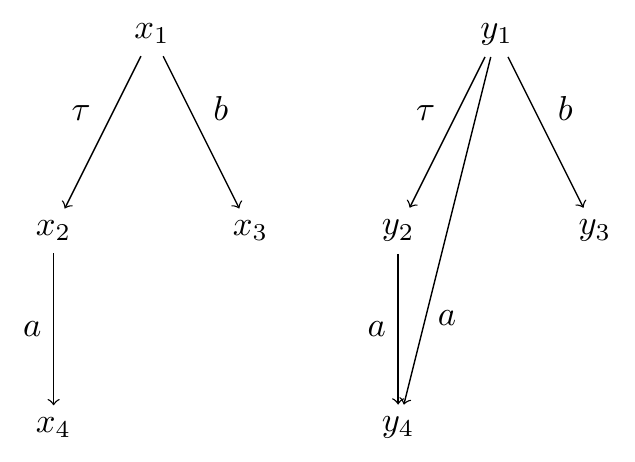}
  \caption{The states $x_1,y_1$ are weakly bisimilar, but not branching bisimilar.}
  \vspace{-0.5cm}
  \label{fig:cex:bbisim}
\end{figure}
Nevertheless, it is well known that the saturation step (i.e. adding strong $a$-transitions for each weak $a$-transition to the transition system) is not sound with respect to branching bisimulation even in the case of labelled transition systems \cite{bbisim:1996}; thus, a different approach to characterise it is required. The reason is that $\tau$-steps can enable or disable choice in an observable behaviour, something that is hidden by the saturation step. For weak bisimulation this point is irrelevant; however, branching bisimulation (which is finer than weak bisimulation) requires that $\tau$-steps which are required to answer an observable action may only lead to states that are still in bisimulation relation with the original state. For instance, consider the states $x_1,y_1$ described in Figure~\ref{fig:cex:bbisim}. They are \emph{not} branching bisimilar because the transition $y_1 \step{a} y_4$ can be simulated by the transitions $x_1 \step {\tau} x_2 \step {a} x_4$, but the intermediate state $x_2$ cannot be related with the state $y_1$ because $y_1$ can fire a $b$-transition which $x_2$ cannot simulate.

Our research hypothesis is that recording executions (up to a notion of stuttering) generated by a state (instead of the set of successor states) is sufficient to capture branching bisimulation across different classes of systems. Stated differently, it is the set of executions (not the set of successor states) which specifies the branching structure of a state in the presence of silent actions. In particular, we will substantiate this claim for the class of labelled transition systems and fully probabilistic systems in this paper.

In lieu of the above hypothesis, we restructure the classical coalgebraic machinery in the following way. We begin by studying a notion of \emph{paths} on an arbitrary set $X$ (denoted $\paths X$) in Section~\ref{sec:prelim}, which is general enough to specify the executions of a labelled transition system and a fully probabilistic system. Intuitively, a path on $X$ can be viewed as a finite sequence that alternates between the elements of $X$ and an action in the alphabet $\tact=\act\uplus\{\tau\}$, where $\tau\not\in\act$ is the silent action. Now for every path $p\in\paths X$ there is a unique stutter invariant path $\stutter p\in\paths X$ associated with it, which intuitively can be constructed by removing the $\tau$ self-loops. This is reminiscent of a coloured trace from \cite{bbisim:1996}, which is obtained from a concrete coloured trace in a process graph whose nodes are labelled by a fixed set of colours. 
In the sequel, stutter invariance induces an equivalence relation $\sim$ on the set $\paths X$ whose quotient is denoted as $\spaths X$. Furthermore, it turns out that both the mappings $\paths X,\spaths X$ are endofunctors on the category of sets $\set$.

Now what is missing in our approach is the type of dynamics (also known as the branching type in the theory of coalgebras).
For instance, a labelled transition system can be viewed as a coalgebra of type $\mcal P\circ (\tact \times \id{})$ over the base category $\set$. Here $\mcal P$ is the covariant powerset functor and $\tact \times \id{}$ is the product functor whose left component is fixed. In other words, the branching type of labelled transition system is nondeterministic. Therefore, to characterise branching bisimulation, we consider coalgebras of type $\mcal P\circ \tspaths$ over the base category $\set$. In Section~\ref{sec:lts}, we show that behavioural equivalence in this coalgebra coincides with the traditional branching bisimulation relation \cite{bbisim:1996}.  Moreover, this framework can also be used to characterise the weak bisimulation, delay bisimulation, and eta-bisimulation relations; however, for reasons of space, this is worked out in Appendix~\ref{sec:weak-eta-del}.

Nevertheless, the situation is not so straightforward in the case of a fully probabilistic system. Often such systems are modelled as coalgebras of type $\mcal D\circ(\tact \times \id{})$ over the base category $\set$, where $\mcal D$ is the sub-distribution functor. It turns out that one needs a notion of measurable space and a measure on the set of maximal executions\footnote{An execution of a fully probabilistic system is \emph{maximal} if it is an infinite execution or it stops in a state with the sum of probabilities of all outgoing transitions as $0$.} in order to define branching (weak) bisimulation relations over the states of a fully probabilistic system (cf. \cite{BaierHermanns:1997,sokolova:2009:sacs}). Thus, it is natural to consider fully probabilistic systems as `weighted' coalgebras of type $\G\circ (\tact\times\id{})$ over the base category of measurable spaces $\meas$. Here, $\G$ is the well-known Giry monad of probability measures \cite{Panangaden:2009:LMP}. 

In Section~\ref{sec:prob}, just like in the discrete case, we consider coalgebras of type $\G \circ \tspaths$ over the base category $\meas$ to characterise probabilistic delay bismulation, which was mistakenly \cite{sokolova-comm} called probabilistic branching bisimulation in \cite{sokolova:2009:sacs,Sokolova:2005:actiontype}. The crux of the matter is in defining $\paths X$ and $\spaths X$ as endofunctors on the category $\meas$. In other words, we need to resolve the following issues: first, which subsets of $\paths X$ and $\spaths X$ are measurable; second, whether $\paths X \rTo^{\paths f} \paths Y$ (for a given $X \rTo^f Y$ in $\set$) is a measurable function or not; third, constructions of measures on the sets $\paths X$ and $\spaths X$. These issues are explored in Section~\ref{sec:pmeas-paths}, for which some preliminary knowledge on topology, domain theory, and measure theory is required. In Section~\ref{sec:conc}, we discuss future directions for research and present some concluding remarks. All the complete proofs pertaining to each section can be found in the appendices~\ref{app-prelim}-\ref{app-prob}.


\section{Preliminaries}\label{sec:prelim}


This section is devoted to formally introduce a notion of path and stutter equivalent path on a set $X$, which will be used throughout the paper. As mentioned earlier, a path on $X$ can be intuitively viewed as a finite sequence that alternates between the elements of the set $X$ and an action in the alphabet $\tact$. However, we abstain from this operational view in favour of Definition~\ref{def:path} to reason about paths from a functional perspective.

Let $\tact^*$ be the set of finite words with $\eseq{}\in\tact^*$ denoting the empty sequence. We write $\preceq$ to denote the prefix ordering on words and let $\history\sigma= \{\sigma' \mid \sigma'\preceq \sigma\}$.

\begin{definition}\label{def:path}
  A \emph{path} $p$ on a set $X$ is a function whose codomain is $X$ and domain is the set of all prefixes of some word in $\tact^*$.
\end{definition}


Let $\paths X$ be the set of all paths on a given set $X$. Then, this lifts to an endofunctor on the category of sets $\set$ by letting: $\paths f(p)=f\circ p$, for every function $X \rTo^f Y$.

Every path $p\in \paths X$ has a \emph{trace} associated with it. Moreover, every path $p\in \paths X$ reaches a \emph{last} element from the set $X$. Symbolically, we write
\[\trace p = \max \dom p\quad \text{and} \quad \last p = p (\trace p).\]
\begin{proposition}\label{prop:trace-last-pres}
  Let $X \rTo^f Y$ be a function. Then, for every path $p\in \paths X$ we have $\trace p = \trace {fp}$ and $f \last p = \last{fp}$.
\end{proposition}

\begin{wrapfigure}{r}{0.32\textwidth}
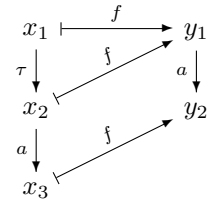

  \centering
  \vspace{-1cm}
  \begin{diagram}[height=1.5em,scriptlabels]
  x_1 & \rMapsto^f & y_1\\
  \dTo^{\tau} & \ruMapsto^f & \dTo^a\\
  x_2 &  & y_2\\
  \dTo^a & \ruMapsto^f &\\
  x_3 & &
\end{diagram}
\vspace{-0.6cm}
\caption{Two systems having same branching structure.}\label{ex:inerttau}
\vspace{-0.5cm}
\end{wrapfigure}
The above proposition states that the functor $\tpaths$ preserves the trace of a path, which is quite strong for our purpose.
To exemplify this, consider two labelled transition systems and a function between the states as shown in Figure~\ref{ex:inerttau}.
Since the $\tau$-step from the state $x_1$ does not disable any choice of observable action offered by $x_1$, we would like to declare the states $x_1$ and $x_2$ as equivalent.
In other words, we would like to assert that $f$ is a homomorphism between the coalgebras $(X=\{x_1,x_2,x_3\},\alpha)$ and $(Y=\{y_1,y_2\},\beta)$, where the functions $\alpha,\beta$ return all the generated executions. However, we note that this is not the case because the $f$-image of the execution $p=\langle x_1\ \tau \ x_2 \ a \ x_3 \rangle \in \alpha(x_1)$ is $\langle y_1 \ \tau\ y_1 \ a \ y_2 \rangle$ which is not an execution from $y_1$. Thus, the key observation is to relate the executions of two systems up to stuttering, which leads to the following definition.

\begin{definition}\label{def:stutter}
 A function $\phi$ with $\dom\phi=\dom p$ and $\cod\phi=A_\tau^*$ is a \emph{stutter basis} for a path $p\in\paths X$ if it can be constructed inductively by the following rules: 
\begin{enumerate}
  \item $\phi(\eseq{})=\eseq{}$.
  \item if $\sigma'\tau\in\dom p$ and $p(\sigma'\tau)=p(\sigma')$ then $\phi(\sigma'\tau)=\phi(\sigma')$.
  \item if $\sigma'\tau\in\dom p$ and $p(\sigma'\tau)\neq p(\sigma')$ then $\phi(\sigma'\tau)= \phi(\sigma')\tau$.
  \item if $\sigma'a\in\dom p$ and $a\in\act$ then $\phi(\sigma' a) = \phi(\sigma') a$.
\end{enumerate}
\end{definition}
As an example, consider a path $p=\langle x_1 \ \tau\ x_1 \ a \ x_2 \rangle$. Then, the function $\phi$ defined as $\phi(\tau)=\phi(\eseq{})=\eseq{}$ and $\phi(\tau a)=a$ is a stutter basis $\phi$ for the path $p$. However, if $p=\langle x_1\ \tau \ x_2 \ a\ x_3 \rangle$ with $x_1\neq x_2$, then $\phi=\id{}$ is a stutter basis for $p$.


\begin{theorem}\label{thm:stutter-unique}
For any path there is a unique stutter basis.
\end{theorem}

\begin{lemma}\label{lemma:stutt-hpres}
  Let $\phi$ be the stutter basis for a path $p\in\paths X$ with $\dom p=\history {\sigma}$, for some $\sigma\in\tact^*$. Then, $\phi(\history {\sigma})=\ \history{\phi(\sigma)}.$
\end{lemma}

\begin{definition}
  Given a path $p\in\paths X$ and its corresponding stutter basis $\phi$, then a function $\phi(\dom p)\rTo^{\stutter p} X$ is the \emph{stutter invariant} path relative to $p$ if $\stutter p \circ \phi =p$.
\end{definition}

Notice that for the function $\stutter p$ to be a path its domain should be a prefix closed subset of a word, which follows directly from Lemma~\ref{lemma:stutt-hpres}.


The notion of stutter invariant path induces an equivalence relation on the set of all paths as follows. Two paths $p,q\in\paths X$ are said to be \emph{stutter equivalent}, denoted $p\sim q$, if and only if they have the identical stutter invariant path, i.e., $\stutter p = \stutter q$. Let $\spaths X$ be the set of all the paths up to stutter equivalence. This lifts to a functor as well:
\begin{equation}\label{eq:spaths}
  \spaths f [p]_\sim =[f\circ p]_\sim \qquad \text{for any $p\in\paths X$ and $X\rTo^f Y$}.
\end{equation}
To prove that the above map is well-defined, we need the following lemma.
\begin{lemma}\label{lemma:stutt-preserve}
  For any $p\in\paths X$ and any $X\rTo^f Y$, we have $f\circ \stutter p \sim f\circ p$.
\end{lemma}
\begin{theorem}\label{thm:spath-functor}
  The mapping in \eqref{eq:spaths} is well defined and $\tspaths$ is an endofunctor on $\set$.
\end{theorem}
\textbf{Notation} We write $\pi_X$ for the quotient map that maps a path $p\in\paths X$ to $\qset{[p]}$.

We end this subsection with few properties on (stutter equivalent) paths. 
\begin{lemma}\label{blemma}
  Let $p\in\paths X,q\in \paths Y$, and $X \rTo^f Y$.
  \begin{enumerate}
    \item\label{blemma:stutt-pres-startstate}If $fp\sim q$ and $\trace q \in \tau^*a\tau^*$ then $q(\eseq{})= fp(\eseq{})\land \trace p \in \tau^*a \tau^*$.
    \item\label{blemma:last-stutter}  $\last p = \last {\stutter p}$.
    \item\label{blemma:stutt-pres-laststate} If $fp \sim q$, then $f(\last p)=\last q$.
    \item\label{blemma:pi-pres} $\pi_Y \circ \paths f = \spaths f \circ \pi_X$.
    \item\label{blemma:inverse} Let $V_i\subseteq \spaths Y$ (for $i\in I$) be a family of pairwise disjoint sets. Then,
  \[\inv{\spaths f} \Big(\bigcup_{i\in I} V_i\Big) \ = \ \bigcup_{i\in I} \inv {\spaths f} (V_i)\enspace.\]
  \end{enumerate}
\end{lemma}
\subsection*{Coalgebraic preliminaries}

\begin{definition}
  Let $\mathbf C$ be a category and let $\mathbf C \rTo^F \mathbf C$ be an endofunctor. An \emph{$F$-coalgebra} over the base category $\mathbf C$ is a tuple $(X,\alpha)$, where $X$ is an object in $\mathbf C$ and $X \rTo^\alpha FX$ is an arrow in $\mathbf C$.
  Given two objects $(X,\alpha)$ and $(Y,\beta)$, an $F$-coalgebra homomorphism is an arrow $X\rTo^f Y$ in $\mathbf C$ such that $Ff\circ \alpha=\beta\circ f$.
\end{definition}

\begin{definition}
  Let $\mathbf C$ be a concrete category over the category of sets $\set$, i.e., there is a faithful functor $\mathbf C \rTo^{|\_|} \set$. Let $(X,\alpha)$ be an $F$-coalgebra over the concrete category $\mathbf C$. Then, two points $x,x'\in |X|$ are said to be \emph{$F$-behaviourally equivalent} if and only if there is an $F$-coalgebra $(Y,\beta)$ and an $F$-coalgebra homomorphism $X \rTo^f Y$ such that $f(x)=f(x')$.
\end{definition}

In Section~\ref{sec:lts}, we will let $\mathbf C=\set$ and $|\_|=\id{}$; however, in Section~\ref{sec:prob}, we will let $\mathbf C=\meas$ and the faithful functor $|\_|$ to be the forgetful functor which forgets the sigma algebras associated with the underlying sets.

\section{Branching bisimulation on labelled transition systems}\label{sec:lts}
The goal is to characterise branching bisimulation of \citeauthor{bbisim:1996} \cite{bbisim:1996} using a coalgebraic approach based on paths as outlined in the introduction.

\begin{definition}
  A \emph{labelled transition system} is a triple $(X,\tact,\rightarrow)$, where $X$ is a set of states, $\tact$ a set of actions, and $\rightarrow \subseteq X \times \tact \times X$ is the so-called \emph{transition relation}.
\end{definition}
As usual, we write $x\step a x'$ and $\steps {}\subseteq X \times \act^* \times X$ to denote an element $(x,a,x')\in\rightarrow$ and the weak reachability relation, respectively. The latter is defined as the smallest relation satisfying the following inference rules: $\frac{ }{x \steps{\eseq{} } x} \qquad \frac{x \steps \sigma x' \ x' \step a x''}{x \steps{\sigma a} x''}.$

\begin{definition}
  Let $(X,\tact,\rightarrow)$ be a labelled transition system. A symmetric relation $R \subseteq X \times X$ is called a \emph{branching bisimulation} relation \cite{bbisim:1996} if and only if for any $x,y,x'\in X$ and $a\in\tact$,
  if $x \step a x' \land x R y$ then
  $
  (x'R y \land a=\tau) \vee
  \exists_{y',y''}\ y \steps {\eseq{}} y' \step a y'' \land x R y' \land x' R y''.
  $
  Two states $x\in X$ and $x'\in X$ are \emph{branching bisimlar} if and only if there exists a branching bisimulation relation $R$ such that $x R x'$.
\end{definition}

Next, we construct a $\mcal P\circ\tspaths$-coalgebra based on paths, where $\mcal P$ is the covariant power set endofunctor on the category of sets $\set$.

An \emph{execution} starting from a state $x\in X$ of a labelled transition system $(X,\tact,\rightarrow)$ is a path $p\in\paths X$ such that $p(\sigma) \step a p(\sigma a)$, for all $\sigma a\in\dom p$. Let $\exec x$ be the set of all executions starting from $x$. Such a transition system can be modelled as a coalgebra $(X,\pi_X\circ \alpha)$, where transition function $\alpha$ is given as:
\[\alpha(x)=
    \{p \mid p\in\exec x \land \trace{p}\in\tau^*a\} \cup
    \{p \mid p\in\exec x \land \trace p\in\tau^*\}.
\]

\begin{remark}
  At this stage, we would like to highlight the distinction between a path and an execution made in this paper. It should be noted that all executions of a system (under investigation) are paths; however, the converse may not be true. This is not unusual because after all the executions of a system are generated on the basis of how behaviour of the system is specified (for instance, by the transition relation in the case of labelled transition systems and by the transition function in the case of fully probabilistic system).
\end{remark}

Next, we state the main result of this section.
\begin{theorem}\label{thm:bisim}
  Let $(X,\tact,\rightarrow)$ be a labelled transition system and $(X,\pi_X\circ \alpha)$ be the corresponding $\mcal P\circ \tspaths$-coalgebra. Then, two states $x,x'\in X$ are branching bisimilar if and only if the states $x,x'$ are $\mcal P\circ \tspaths$-behaviourally equivalent.
\end{theorem}
\begin{proof}
  \fbox{$\Rightarrow$} Let $R\subseteq X \times X$ be the largest branching bisimulation on the given labelled transition system. Then, from \cite{bbisim:1996} we know that $R$ is an equivalence relation. So let $X \rTo^f X/R$ be the quotient map. Now to show that $f$ is indeed the required $\mcal P\circ \tspaths$-coalgebra homomorphism, we first construct a coalgebra $X/R\rTo^\beta \mcal P\spaths{X/R}$: 
  \[\beta(f(x))=\{\spaths f(p)\mid p\in \alpha(x)\}, \qquad \text{for all}\ x\in X\enspace.\]
  Clearly, $\beta$ is a total function because $f$ is surjective. Next, we claim that $\beta$ is well-defined, i.e., independent of the chosen representative. Let $x,x'\in X$ such that $f(x)=f(x')$. Then, we need to show that $\beta(f(x))=\beta(f(x'))$. Suppose $\qset {[fp]}\in\beta(f(x))$ with $p\in\alpha(x)$. Then, by structural induction on the word $\sigma\in\dom p$ we show that there is a path $p'\in\alpha(x')$ such that $f\circ(p|_{\sigma})\sim f\circ p'$. Here, we write $p|_\sigma$ to denote the restriction of the function $p$ to the sub-domain $\history\sigma$. To see this, without loss of generality, let $\sigma a\in\dom p$. Then by the induction hypothesis we find an execution $p'\in\alpha(x')$ such that $f\circ(p|_\sigma)\sim f\circ p'$. Note that $p(\sigma) \step a p(\sigma a)$ and using Lemma~\ref{blemma}\eqref{blemma:stutt-pres-laststate} we get $f\circ(p|_\sigma)\sim f\circ p' \implies p(\sigma)\ R\ \last {p'}$. Let $a\in \act$. Then, using the transfer property of branching bisimulation we get $\last{p'} \steps {\eseq{}} y \step {a} y'$ such that $p(\sigma) R y$ and $p(\sigma a) R y'$ since $p(\sigma)$ and $\last {p'}$ are branching bisimilar. Moreover, from the stuttering lemma \cite{bbisim:1996} we know that any intermediate state visited in the path $\last{p'} \steps {\eseq{}} y$ is also $R$-related to $p(\sigma)$. Therefore, there is a path $p''\in\paths X$ which extends $p'$ such that $f\circ(p|_{\sigma a})\sim f\circ p''$. In addition, if $a=\tau$ then we either have $p(\sigma \tau)\ R\ \last {p'}$ or $\last {p'} \steps {\eseq{}} y \step \tau y'$, for some $y,y'$, with $p(\sigma) R y$ and $p(\sigma \tau )Ry'$. Suppose the former is true, then clearly we have $f\circ(p|_{\sigma \tau})\sim f\circ p'$. The latter case is similar to the case when $a\in\act$. Thus, for every $p\in\alpha (x)$ there is a path $p'\in\alpha(x')$ such that $f\circ p\sim f\circ p'$. Likewise, we can show the symmetric property when the role of $x$ and $x'$ is interchanged. This completes the proof of the above claim. Clearly, we have $\beta\circ f =\mcal P\spaths f \circ \alpha$.

  \fbox{$\Leftarrow$} Let $(Y,\beta)$ be a $\mcal P\circ \tspaths$-coalgebra and $X \rTo^f Y$ be a $\mcal P\circ \tspaths$-coalgebra homomorphism. 
  Below we rather illustrate why the relation $ x R x' \iff f(x)=f(x')$ is a witnessing branching bisimulation. The proof can be found in Theorem~\ref{athm:bisim}.
\end{proof}

Consider the two labelled transition systems drawn below enclosed inside the two rectangles. Here, $X \rTo^\alpha \mcal P\spaths X$ and $Y \rTo^\beta \mcal P\spaths Y$ denote the corresponding path-based
\begin{center}
\begin{tikzpicture}
\node[draw=black] (A){
  \begin{tikzpicture}
    \node (x1) {$x_1$};
    \node (a) at ($(x1.center)+(-1.2,0.2)$) {$\alpha:$};
    \node (x2) at ($(x1.center)+(-0.5,-1)$) {$x_2$};
    \node (x3) at ($(x1.center)+(0.5,-1)$) {$x_3$};
    \node (x4) at ($(x2.center)+(-0.5,-1)$) {$x_4$};
    \node (x5) at ($(x2.center)+(0.5,-1)$) {$x_5$};
    \node (x1p) at ($(x1.center)+(2,-0.5)$) {$x_1'$};
    \node (x2p) at ($(x1p.center)+(-0.5,-1)$) {$x_2'$};
    \node (x3p) at ($(x1p.center)+(0.5,-1)$) {$x_3'$};
    \path[->]
        (x1) edge node[above left]{$\tau$} (x2)
        (x1) edge node[above right]{$a$} (x3)
        (x2) edge node[above left]{$a$} (x4)
        (x2) edge node[above right]{$b$} (x5)
        (x1p) edge node[above left]{$a$} (x2p)
        (x1p) edge node[above right]{$b$} (x3p);
  \end{tikzpicture}};
\node[draw=black] (B) at ($(A.center)+(5,0)$) {
  \begin{tikzpicture}
    \node (y1) {$y_1$};
    \node (a) at ($(x1.center)+(-1.2,0.2)$) {$\beta:$};
    \node (y2) at ($(y1.center)+(0,-1)$) {$y_2$};
    \path[->]
        (y1) edge[bend left] node[right]{$b$} (y2)
        (y1) edge[bend right] node[left]{$a$} (y2);
  \end{tikzpicture}};
\end{tikzpicture}
\end{center}
coalgebras with $X=\{x_i,x_j'\mid i\in\{1,2,3,4,5\},j\in\{1,2,3\}\}$ and $Y=\{y_1,y_2\}$. Furthermore, let $X \rTo^f Y$ be a function defined as $f(x)=y_1$ if $x\in\{x_1,x_1',x_2\}$; otherwise $f(x)=y_2$. To illustrate why $R$ (as defined above) is a witnessing branching bisimulation, consider the transition $x_1' \step b x_3'$ and $x_1 R x_1'$. Clearly, $\langle x_1' \ b \ x_3' \rangle \in \alpha(x_1)$, which further implies that $\langle y_1 \ b \ y_2 \rangle  \in \beta(y_1)$. Since $\mcal P\spaths f\circ \alpha=\beta \circ f$ we know that there is an execution $p$ such that $f\circ p$ is stutter equivalent to  $\langle y_1 \ b \ y_2 \rangle$. And by inspection we note that $p=\langle x_1 \ \tau \ x_2 \ b \ x_5\rangle$ is such an execution. Moreover, $x_2 R y_1$ and $x_5 R y_2$ which is required by the transfer property of a branching bisimulation relation.

In hindsight, using the terminology of \cite{vanGlabbeek:2001:btime}, a $\mcal P\tspaths$-coalgebra homomorphism preserves the branching structure of states. As a consequence, two behaviourally equivalent states have the same set of executions under the image of a $\mcal P\tspaths$-coalgebra homomorphism up to stutter invariance. For instance, in the above example, the sets of all executions having trace $\tau^*a$ from the states $x_1$ and $x_1'$ are $\{\langle x_1 \ a \ x_3 \rangle, \langle x_1 \ \tau \ x_2 \ a \ x_4 \rangle \}$ and $\{\langle x_1' \ a \ x_2'\rangle\}$, respectively. Notice that the $f$-image of these two sets are equivalent up to stutter invariance. A similar argument can be observed for the set of executions from $x_1,x_1'$ having trace $\tau^*b$. 

Though we have focussed on branching bisimulation, this approach can also be used to capture weak, $\eta$ and delay bisimulation, by defining $\alpha$ differently, saturating $\tau$ leading transitions, trailing $\tau$ transitions or both, respectively. This is made explicit in Appendix~\ref{sec:weak-eta-del}.

\section{A measurable space on paths}\label{sec:pmeas-paths}
As mentioned in the introduction, we will consider coalgebras of type $\mcal G \circ \tspaths$ over the base category $\meas$ to characterise probabilistic delay bisimulation. However, before we do so, we have to fix which subsets of the sets $\paths X$ and $\spaths X$ are measurable together with the construction of a measure on the space of paths, which can be a challenging issue in its own right. In this section, we resolve these fundamental issues by first recalling some basic definitions of measure theory taken from \cite{Panangaden:2009:LMP}.
\begin{definition}
  A set $\Sigma_X\subseteq \mcal P(X)$ of subsets of $X$ is a \emph{sigma-algebra} on $X$ if and only if $X\in \Sigma_X$ and $\Sigma_X$ is closed under the set complements and countable unions. Then, the tuple $(X,\Sigma_X)$ is called a \emph{measurable space}.
  A \emph{measure} space is a measurable space $(X,\Sigma_X)$ with a measure $\Sigma_X \rTo^{\mu_X} \left[0,\infty\right]$, i.e., $\mu_X$ is a function satisfying $\mu(\emptyset)=0$ and the \emph{sigma-additivity} property: for any countable family of pairwise disjoint sets $U_i\in \Sigma_X$ (for $i\in I$) we have
  \[\mu_X(\bigcup_{i\in I} U_i)=\sum_{i\in I} \mu_X(U_i)\enspace.\]
  A \emph{probability} space $(X,\Sigma_X,\mu_X)$  is a measure space with $\mu_X(X)=1$.
  A \emph{discrete} space is a measure space such that $X$ is countable and $\Sigma_X=\mcal P(X)$.
\end{definition}

Here, the arbitrary sum of a family $\{r_i\mid i\in I\}$ of nonnegative real numbers is defined as $\sum_{i\in I} r_i =\sup \{\sum_{i\in J} r_i \mid J\fsubseteq I \}$ (cf.  \cite{sokolova:2009:sacs}), where $J \fsubseteq I \iff J \subseteq I \land  \text{$J$ is a finite set}$.

We want to endow a notion of measurability on the set $\spaths X$; however, for simplicity we first restrict ourselves to the set of all paths on $X$, i.e., $\paths X$.
It turns out that the set of all paths carries a topological structure (precisely, they form what is known as Alexandroff topology \cite{alexandroff:1937}) and also satisfies the so-called Kolmogorov separability axiom. Once we have a topological space, the convention is to consider the smallest sigma-algebra generated by the set of all open sets (also known as the Borel sigma-algebra) as the set of measurable sets.

\begin{definition}
  A \emph{topology} on a set $X$ consists of a set of open sets $\mcal O_X \subseteq \mcal P(X)$ such that: first, the empty set and the whole space are in $\mcal O_X$; second, the set $\mcal O_X$ is closed under finite intersection and arbitrary unions.
  A topological space $(X,\mcal O_X)$ is an \emph{Alexandroff} space if the set $\mcal O_X$ is closed under arbitrary intersection. A topological space $(X,\mcal O_X)$ satisfies the \emph{Kolmogorov separability axiom} ($X$ is a $T_0$ space) if any two distinct points are topologically distinguishable, i.e.,
  $\forall_{x,x'\in X}\ x\neq x' \implies \exists_{U\in\mcal O_X}\ (x\in U \land x'\not\in U) \vee (x\not\in U \land x'\in U).$
\end{definition}

It is well-known that the set of all upward closed subsets generated by a poset forms a $T_0$ Alexandroff space. In particular, our set of paths $\paths X$ carries the following order:
\[p\preceq q \iff \dom p \subseteq \dom q \land \forall_{\sigma\in\dom p}\ q(\sigma)=p(\sigma)\enspace.\]
Actually, the above ordering is a prefix order in the sense of \citeauthor{Cuijpers:2013:DCM} \cite{Cuijpers:2013:DCM}.
\begin{definition}
  A \emph{prefix order} is a partial order whose every principal ideal is a totally ordered set.
\end{definition}
\begin{proposition}
  The history of a path $p\in\paths X$ is downward total, i.e., the set $\history{p}=\{p' \mid p'\preceq p\}$ is a totally ordered set.
\end{proposition}
\begin{proposition}
  The set of all paths $\paths X$ on a set $X$ forms a $T_0$ Alexandroff space, whose open sets are upward closed subsets of $\paths X$, i.e., $\mcal O_{\paths X}=\{U \subseteq \paths X \mid U=\future U \}$. Here, the set $\future U=\{p' \mid \exists\ p\in U \land p\preceq p'\}$ denotes future of paths in the set $U$.
\end{proposition}
At this stage, we note the following relationship between stutter paths and the order $\preceq$. 
\begin{lemma}\label{lemma:stutt-hreflect}
  Let $X$ be a set. Then we have the following property: for any two paths $p_1,p_2\in\paths X$, if $\stutter p_1 \preceq \stutter p_2$ then $\exists_{p\in\paths X}\ p \sim p_1 \land p\preceq p_2$.
\end{lemma}
Every point $x\in X$ in an Alexandroff space has a special neighbourhood associated with it, often called the \emph{smallest neighbourhood} of $x$, denoted $\mcal N(x)=\bigcap \{U\mid U\in\mcal O_X\land x\in U\}$. In particular, this structure, in the case of paths, is the principal filter generated by a path. 
\begin{proposition}
  For a path $p\in\paths X$, the smallest neighbourhood of $p$ represents the future of the path $p$, i.e., $\future p=\mcal N(p)$. In contrast, the closure $\closure{p}$ of a path $p\in\paths X$ -- the smallest closed set that contains $p$ -- represents the history of $p$, i.e., $\closure p=\history p$.
\end{proposition}
The next proposition states that the subsets of paths which belong to the Borel sigma-algebra $\mcal B(\mcal O_{\paths X}))$ are measurable.
\begin{proposition}
The tuple $(\paths X, \Sigma_{\paths X})$, where $\Sigma_{\paths X}=\mcal B(\mcal O_{\paths X})$, is a measurable space. Here, $\mcal B(\mcal X)$ denotes the smallest sigma-algebra generated by $\mcal X \subseteq \mcal P(X)$.
\end{proposition}

Next, we establish that the $\paths X \rTo^{\paths f} \paths Y$ (for a given $X \rTo^f Y$) is measurable, i.e., if $V\in\Sigma_{\paths Y}$ then $\inv fV\in\Sigma_{\paths X}$. For this, we need the following result.
\begin{theorem}\label{thm:historypres}
  For any $X \rTo^f Y$, the function $\paths f$ is an order embedding, i.e., for any $p,p'\in\paths X$ we have $p \preceq p' \iff f\circ p \preceq f\circ p'$.
\end{theorem}
Since every order preserving function is continuous and every continuous function is Borel measurable, it follows that, in particular, $\paths f$ is Borel measurable.
\begin{corollary}\label{cor:paths-measurable}
  For any $X \rTo^f Y$, the function $\paths f$ is measurable.
\end{corollary}
In hindsight, the function $\paths f$ is an arrow in the category $\meas$.
Next, we turn our attention on constructing a measurable space on the set $\spaths X$. The idea is to first define an order on the quotient space $\spaths X$, which can be inherited from the underlying space of paths $\paths X$ by simply letting:
$\qset {[p]} \preceq \qset {[q]} \iff \stutter p \preceq \stutter q$, for all $p,q\in\paths X$.

\begin{lemma}\label{lemma:prec-qset-welldef}
  The relation $\preceq$ on the set $\spaths X$ is a well-defined partial order. Furthermore, the relation $\preceq$ on the set $\spaths X$ is also a prefix order.
\end{lemma}

Once we have an order on the quotient space, we can establish that the quotient maps are order preserving (or continuous in the topological sense).
\begin{theorem}\label{thm:quotientmea}
  The quotient function $\paths X \rTo^\pi \spaths X$ is order preserving. Consequently, the quotient function $\paths X \rTo^\pi \spaths X$ is Borel measurable, where the sigma-algebra on paths is given by $\Sigma_{\spaths X}=\mcal B(\mcal O_{\spaths X})$.
\end{theorem}
Next, we state the main theorem of this section.
\begin{theorem}\label{thm:spath-opres}
  For any $X \rTo^f Y$, the function $\spaths f$ is order preserving. Thus, the function $\spaths f$ is Borel measurable.
\end{theorem}

\subsection*{Constructing measures on the space of paths}
Often, measures on a space are constructed in a top-down manner by identifying a measurable set of building blocks and defining a set-function on this collection (for example, in the case of Lebesgue measures on $\mathbb R$, a semi-closed interval $[r,r')$ with $r\leq r'$ is one such building block and the set-function maps every interval of the form $[r,r')$ to the value $r'-r$). In turn, measure extension theorems (for instance, the well-known Carath\'{e}odory-Hahn extension theorem; see \cite[pp~356]{royden:real-measure}) are invoked to lift the set-function on building blocks to a measure on the whole measurable space. In this paper, we will follow a similar recipe; our building blocks will be open subsets of paths. As for measure extension theorems, we will use a result (cf. Theorem~\ref{thm:alvarez}) established by \citeauthor{alvarez-manilla2002} \cite{alvarez-manilla2002}. Below, we recall some definitions on a topological space necessary to state this result.
\begin{definition}
  Let $(X,\mcal O_X)$ be a topological space. A function $\mcal O_X \rTo^\mu \nnreals$ is a \emph{valuation} if and only if the following conditions are satisfied.
  \begin{enumerate}
    \item The function $\mu$ is strict, i.e., $\mu(\emptyset)=0$
    \item The function $\mu$ is order preserving, i.e., for any two open sets $U,U'\in\mcal O_X$, we have $U\subseteq U'$ implies $\mu(U) \leq \mu(U')$.
    \item The function $\mu$ is modular, i.e., for any two open sets $U,U'\in\mcal O_X$, we have $\mu(U)+\mu(U')=\mu(U\cup U') + \mu(U \cap U')$.
  \end{enumerate}
  A valuation $\mu$ is \emph{Scott-continuous} if and only if for any directed family of open sets $(U_i)_{i\in I}$ we have $\mu(\bigcup_{i\in I} U_i) = \sup_{i\in I} \mu(U_i)$. Lastly, a valuation $\mu$ is \emph{locally finite} if and only if every point has a finitely valued open neighbourhood.
\end{definition}
\begin{definition}
  A space $(X,\mcal O_X)$ is \emph{locally compact} if and only if
  for every point $x$ and open set $U$ with $x\in U$, there is a compact subset $V\subseteq X$ such that $x\in \mathrm{int}(V)$ and $V\subseteq U$. Here, $\mathrm{int}(V)$ denotes the interior of $V\subseteq X$. 
\end{definition}

\begin{definition}
  A topological space $(X,\mcal O_X)$ is \emph{sober} if and only if every irreducible closed set is a closure of a unique point. A closed set $C$ is \emph{irreducible} if and only if $C$ is nonempty and it cannot be expressed as union of two smaller closed subsets, i.e., if $C=C_1 \cup C_2$ and $C_1, C_2$ are closed sets, then $C=C_1$ or $C=C_2$.

  We call a subset $C\subseteq X$ \emph{non-sober} if $C$ is irreducible, $C$ is closed, and it cannot be stated as a closure of point (i.e., $\nexists_{x\in X}\ C=\closure x$).
\end{definition}

\begin{theorem}[\cite{alvarez-manilla2002}]\label{thm:alvarez}
  Every locally finite and Scott-continuous valuation on a locally compact sober space extends uniquely to a Borel measure.
\end{theorem}
The restrictions on $\mcal O_X \rTo^\mu \nnreals$ imposed by the above theorem are not unreasonable; at-least for our purpose. In Section~\ref{sec:prob}, we will construct a locally finite and a Scott-continuous valuation on open subsets of paths, which is induced by a given fully-probabilistic transition system. Nevertheless, we cannot immediately apply Theorem~\ref{thm:alvarez} because our space $\paths X$ is not a sober space, even though it is locally compact, i.e., every path $p\in\paths X$ has a compact neighbourhood (since $p\in \future p$). As a result, in the following, we first `soberify' our space $\paths X$ and use Theorem~\ref{thm:alvarez} to construct a Borel measure on $\paths X$ by lifting a given locally finite and Scott-continuous valuation $\mcal O_{\paths X} \rTo^\mu \nnreals$.

\begin{remark}
By inspection, we note that our space $\paths X$ is non-sober. For instance, if $X$ is non-empty then unfolding a $\tau$-loop results in an infinite chain of paths without any maximum since the domain of a path is a set of prefixes generated by some \emph{finite} word.
\end{remark}
Recall that, for a set $X$, both sets of paths $\paths X$ and stutter-equivalent paths $\spaths{X}$ are prefix orders. We want to construct measures on both kinds of spaces, therefore below we work with a class of \emph{simple} prefix orders which generalises both the structures.
\begin{definition}A prefix order is \emph{simple} if the history of every point is a finite set.\end{definition}

For example, the sets $\paths X$ and $\spaths X$ are simple prefix orders.
\begin{proposition}\label{prop:irr-downclosed}
  A directed subset of a prefix order is always totally ordered. In addition, an irreducible downward closed subset of a prefix order is always totally ordered. 
\end{proposition}
Next, we construct a space $X^\infty$ consisting of all points from $X$ in which the non-sober sets (w.r.t. Alexandroff topology) are added as limit points.
\begin{align*}
  X^\infty =& \ X \cup \{\infty_C \mid \text{$C\subseteq X$ is a non-sober set w.r.t. Alexandroff topology}\}.\\
  \preceq'=& \preceq \cup  \ \{(\infty_C,\infty_C) \mid \infty_C \in X^\infty\} \cup \{(x,\infty_C) \mid x\in C\}.
\end{align*}
As an example, consider the prefix order $(\mathbb N,\leq)$ with their natural ordering. The sober space $\mathbb N^\infty=\mathbb N\cup \{\infty_{\mathbb N}\}$ is isomorphic to the well-known set of extended natural numbers $\mathbb N_\omega$.
\begin{lemma}\label{lemma:dtotal-lpaths}
  The set $X^\infty$ is prefix ordered by the relation $\preceq '$, if $(X,\preceq)$ is a prefix order.
\end{lemma}
Henceforth, we do not distinguish between the relation $\preceq$ and $\preceq'$. Notice that being sober is a topological property and therefore, we need a `right' notion of topology on $X^\infty$ to qualify it as sober. For instance, if we take upward closed sets as open sets (just like in the case of $X$) we find that the space $X^\infty$ is still non-sober; as a result, $X^\infty$ is non-sober w.r.t. Alexandroff topology. However, if we endow $X^\infty$ with a Scott topology then the space becomes sober w.r.t. this finer topology. For example, in the case of extended natural numbers, the problematic case of the directed set $\mathbb N$ (which was non-sober w.r.t. Alexandroff topology) is actually not a Scott-closed set\footnote{A subset $C\subseteq X$ of a prefix order $(X,\preceq)$ is Scott closed if and only if $C$ is downward closed and for any directed set $D\subseteq C$, if $\sup D$ exists then $\sup D \in C$.} since $\sup \mathbb N=\infty_{\mathbb N}$ and $\infty_{\mathbb N}\not\in\mathbb N$.
\begin{proposition}\label{prop:scott-sober}
Let $(X,\preceq)$ be a simple prefix order. A subset $U\subseteq X^\infty$ is \emph{Scott} open if and only if $U$ is upward closed and it is inaccessible by directed joins, i.e., for any directed set $D\subseteq X^\infty$ if $\sup D$ exists and $\sup D\in U$ then $D \cap U \neq \emptyset$. Let $\mcal S_{X^\infty}$ denote the collection of Scott open subsets of $X^\infty$. Then, the space $(X^\infty,\mcal S_{X^\infty})$ is a sober space.
\end{proposition}
  To apply Theorem~\ref{thm:alvarez}, we need to first construct a locally finite and Scott-continuous valuation on our new sober space $\lpaths X$. In the following theorem, we will construct one such valuation on $\lpaths X$ from an old valuation $\mcal O_{\paths X} \rTo^\mu \nnreals$.
\begin{theorem}\label{thm:val-scott-extension}
  Let $(X,\preceq)$ be a prefix order and let $\mcal O_{X} \rTo^\mu \nnreals$ be a locally finite and Scott-continuous valuation. Then, the function $\mcal S_{X^\infty} \rTo^{\tilde \mu} \nnreals$ defined as follows:
  \[\tilde \mu(V) = \mu(V \cap X) \quad \text{(for every Scott-open set $V\in \mcal S_{X^\infty}$)}\]
  is a Scott-continuous valuation. If $X$ is simple then $\tilde\mu$ is locally finite.

\end{theorem}
As a result, the function $\tilde\mu$ lifts to a unique Borel measure on the sigma-algebra $\Sigma_{\lpaths X}=\mcal B(\mcal S_{\lpaths X})$ due to Theorem~\ref{thm:alvarez}. However, in order to reflect back this measure on the original sigma-algebra $\Sigma_{\paths X}$, it is sufficient to establish that $\Sigma_{\paths X}$ is contained in the Borel sigma-algebra induced by Scott-open sets, i.e., $\Sigma_{\paths X} \subseteq \Sigma_{\lpaths X}=\mcal B(\mcal S_{\lpaths X})$. The next theorem states under what conditions the set-containment between the two sigma-algebras $\Sigma_{\paths X},\Sigma_{\lpaths X}$ is possible.
\begin{theorem}\label{thm:bs-contained}
  For countable sets $A$ and $X$, the sigma-algebra $\Sigma_{\paths X}$ is contained in the Borel sigma-algebra $\Sigma_{\lpaths X}$. Moreover, we also have $\Sigma_{\spaths X} \subseteq \Sigma_{\slpaths{X}}$.
\end{theorem}

\begin{corollary}\label{cor:meas-gen}
  Suppose the sets $A$ and $X$ are countable. Then every locally finite and Scott-continuous valuations $\mcal O_{\paths X} \rTo^\mu \nnreals$ and $\mcal O_{\spaths X} \rTo^{\mu} \nnreals$ lifts to a unique Borel measure $\Sigma_{\paths X} \rTo^{\tilde \mu} \nnreals$ and $\Sigma_{\spaths{X}} \rTo^{\tilde{\mu}} \nnreals$, respectively.
\end{corollary}
\section{Probabilistic delay  bisimulation}\label{sec:prob}

In this section, we use the concepts developed in the previous section to characterise probabilistic delay bisimulation relation between the states of a fully probabilistic system.

\begin{definition}
  A \emph{(fully) probabilistic transition system} is a triple $(X,\tact,P)$ consisting of a countable set of states $X$, a countable set of actions $\tact$, and a \emph{probability transition function} $X \times \tact \times X \rTo^P \left[0,1\right]$ such that for every $x\in X$, the set $\{(a,x')\mid 0 < P(x,a,x')\}$ is finite and $\sum_{(a,x')\in \tact \times X} P(x,a,x')\in\{0,1\}$.
\end{definition}

Given a probabilistic transition system $(X,\tact,P)$, an \emph{execution} $p$ is a path on $X$ such that
$
    \forall_{\sigma a\in \dom p}\ 0<P(p(\sigma),a,p(\sigma a))
.$
Let $\exec x$ be the set of all executions starting from the state $x$. 
We write $\hat a=\eseq{}$ if $a=\tau$ and $\hat a=a$ if $a\in\act$.
\begin{definition}
  An equivalence relation $R \subseteq X \times X$ on a probabilistic transition system $(X,\tact,P)$ is a \emph{probabilistic delay bisimulation} \cite{BaierHermanns:1997,sokolova:2009:sacs} if and only if 
  \[\forall_{x,x'\in X}\ x R x' \implies \forall_{x'' \in X,a\in\tact}\ \ P(x,\tau^*\hat a,[x'']_R)= P(x',\tau^*\hat a,[x'']_R)\enspace.\]
  Two states $x,x'\in X$ are \emph{probabilistic delay bisimilar} if and only if there is a probabilistic delay bisimulation $R$ such that $xR x'$.

  Here, the probabilities associated with weak transitions are defined (originally given in \cite{Sokolova:2005:actiontype}) in the following way. For $x\in X,Y\subseteq X,L\subseteq \tact^*$, we let
  \begin{itemize}
    \item $\begin{aligned}[t]
            x \step{L}Y =&\ \{p\in\exec x \mid \trace p\in L\land \last p\in Y\ \land \\
            &\ \qquad \forall_{q}\ (q\prec p \land \trace q\in L ) \implies \last q\not\in Y\}.
        \end{aligned}$
    \item  $P(x,L,Y)= \sum_{p\in x \step{L} Y} \mu_P(p)$, where $\paths X \rTo^{\mu_P} \left[0,1\right]$ is defined as:
\[
\mu_P(p)=
\begin{cases}
    \prod_{\sigma a\in\dom p} P(p(\sigma),a,p(\sigma a)), & \text{if } p\in\exec x,\ \text{for some $x\in X$} \\
    0, & \text{otherwise}
  \end{cases}.
\]
  \end{itemize}
\end{definition}
\begin{proposition}\label{prop:ord-reversal}
For a given fully probabilistic system $(X,\tact,P)$, the induced function $\mu_P$ on paths is order reversing. Moreover, $\mu_P(\eseq{x})=1$ (for any $x\in X$).
\end{proposition}

In contrast to Section~\ref{sec:lts}, our base category will be rather the category of measurable spaces and measurable functions $\meas$.
\begin{definition}
  Below we recall the well known Giry functor $\G$ (see e.g. \cite{Panangaden:2009:LMP}):
  \begin{itemize}
    \item Let $(X,\Sigma_X)$ be a measurable space. Then, $\G(X,\Sigma_X)=(\G X,\Sigma_{\G X})$, where $\G X$ is the set of all probability measures on the measurable space $(X,\Sigma_X)$. The sigma-algebra $\Sigma_{\G X}$ is the smallest sigma-algebra such that the evaluation maps $\G X \rTo^{\eval U} \left[0,1\right]$ are Borel measurable, for every $U\in \Sigma_X$.
    \item For any arrow $X \rTo^f Y$ in $\meas$, we let $\G(f)(\mu)=\mu\circ \inv f$.
  \end{itemize}
\end{definition}

\begin{wrapfigure}{r}{0.3\textwidth}
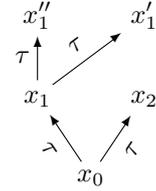

  \centering
  \vspace{-1.3cm}
  \begin{diagram}[height=1.5em,width=2em]
  x_1'' & & x_1'\\
  \uTo^{\tau}& \ruTo^{\tau} &\\
  x_1 & & x_2\\
  & \luTo(1,2)_{\tau} \ruTo(1,2)_{\tau}& \\
  & x_0 &
\end{diagram}
\vspace{-0.5cm}
\caption{An example motivating separation closure.}
\label{Fig:motive}
    \vspace{-0.5cm}
\end{wrapfigure}
To motivate the next definition, consider the transition system depicted in Figure~\ref{Fig:motive} and assume a nonzero probability with each of the drawn transitions. Furthermore, let $p_1,p_1'$, and $p_2$ be the three executions that reach the states $x_1,x_1'$, and $x_2$, resp., from the state $x_0$.
Notice that $P(x_0,\tau^*,\{x_1,x_1',x_2\})$ is the sum of the probabilities associated only with the executions $p_1,p_2$.
The execution $p_1'$ is not considered in the above computation because one can reach the set of target states $\{x_1,x_1',x_2\}$ with the execution $p_1$ which is a prefix of $p_1'$. This means, the execution $p_1'$ is redundant and neglected while computing the probability to reach the above target states.
Such redundancies at the level of paths are identified by the following notion of separation closure.
\begin{definition}
  Let $(X,\mcal O_X)$ be an Alexandroff space. The \emph{separation closure} of a subset $U \subseteq X$ is the set
  $U^\star = \{x\in U \mid \closure x \cap U = \{x\}\}.$ A subset $U\subseteq X$ is \emph{separated} if $U=U^\star$.
\end{definition}

In the context of the previous example (Figure~\ref{Fig:motive}), let $U=\{p_1,p_1',p_2\}$. Then, we find that $\history x \cap U=\{x\}$, for $x\in \{p_1,p_2\}$, while for the execution $p_1'$ we find that $\history{p_1'}\cap U= \{p_1,p_2\}$. Thus, $U^\star=\{p_1,p_2\}$ which were the only executions needed to compute the probability to reach one of the target states. Incidentally, a separated subset of paths $U\subseteq \paths X$ (i.e., $U=U^\star$) is minimal in the sense that any two distinct paths $p,p'\in U$ are not in the prefix relation $\preceq$, i.e., $p\not\preceq p'$ and $p'\not\preceq p$. The following proposition asserts this.
\begin{proposition}\label{prop:separated}
  In an Alexandroff space $(X,\mcal O_X)$, a separated subset $U\subseteq X$ (i.e., $U=U^\star$) has topologically distinguishable points. Moreover, 
  \begin{enumerate}
    \item the separation closure of a set is always separated, i.e., $U^\star = {U^\star}^\star$, for any $U\subseteq X$.
    \item\label{prop:separated:hereditary} the collection of separated sets is hereditary, i.e., if $U_1\subseteq U_2$ and $U_2=U_2^\star$, then $U_1=U_1^\star$.
    \item\label{prop:separated:future} for any subset $U\subseteq X$, we have $(\future U)^\star\subseteq U^\star$. Moreover the converse also holds, if the underlying space $X$ is a $T_0$ space. Here, upward closure is w.r.t. the specialisation order $\preceq$, i.e., $x \preceq x' \iff \closure x \subseteq \closure{x'}$, for any $x,x'\in X$.
  \end{enumerate}
\end{proposition}
Thus, separation closure provides an alternative way to compute $P(x,L,Y)$.
\begin{lemma}\label{lemma:step-sepclosure}
  For a given system $(X,\tact,P)$, define the set
  $x \steps L Y =\{p\in\paths X \mid p(\eseq{})=x\land \trace p\in L\land \last p\in Y \}.$
  Then, $\sum_{p\in x \step L Y}\mu_P(p)= \sum_{p\in (x \steps L Y)^\star} \mu_P(p) $.
\end{lemma}
It should be noted that for a given probabilistic transition system $(X,\tact,P)$, we have $x \step L Y\subseteq (x \steps L Y)^\star$; however, the converse is not true in general.

The next theorem highlights a property characteristic to fully probabilistic systems. It states that if a separated subset of paths $U$ has a lower bound $p$, i.e., $\forall_{q\in U}\ p\preceq q$, then the sum of probabilities associated with each path in $U$ is bounded by the weight of $p$. This property is due to the order reversing nature of the function $\mu_P$ (cf. Proposition~\ref{prop:ord-reversal}).
\begin{theorem}\label{thm:sep:prob-char}
Given a system $(X,\tact,P)$, a path $p\in \paths X$, and a separated set of paths $U\subseteq\paths X$ such that $p\preceq U$, i.e., $\forall_{q\in U}\ p\preceq q$. Then, $\sum_{q\in U}\mu_P(q)\leq\mu_P(p)$.
\end{theorem}
Next we focus on the construction of probability measures on the space $\paths X$. From Corollary~\ref{cor:meas-gen}, it suffices to construct a locally-finite and Scott-continuous valuation on the open subsets of paths. The following theorem extends the function $\mu_P$ (induced by a given fully probabilistic system $(X,\tact,P)$) to such a valuation on paths.
\begin{theorem}\label{thm:pts-valuation}
  Given a system $(X,\tact,P)$, then the function $\mcal O_{\paths X} \rTo^{\tilde\mu_P} \nnreals$ defined as $\tilde\mu_P(U)=\mu_P(U^\star)$ (for every open set $U$) is a locally finite and Scott-continuous valuation.
\end{theorem}

Now we have all the technical machinery to encode a given probabilistic transition system $(X,\tact,P)$ as a coalgebra $X\rTo^\alpha \mcal G\spaths X$, where $\Sigma_X=\mcal P(X)$ (since $X$ is countable). The transition system $\alpha$ is defined, coalgebraically, as follows:
\begin{equation}\label{eq:measure-alpha}
  \alpha(x)(U)= \tilde\mu_P(\inv {\pi_X}(U) \ \cap \future{\eseq{x}}), \qquad \text{for every}\ U\in\Sigma_{\spaths X}\enspace.
\end{equation}
Here, we abuse notation by using $\tilde\mu_P$ to denote a measure on $\Sigma_{\paths X}$. Note that this measure is rather constructed by extending the valuation given in Theorem~\ref{thm:pts-valuation}. 
\begin{proposition}\label{prop:alpha-measure}
  The mapping in \eqref{eq:measure-alpha} is a probability measure.
\end{proposition}
Now we are ready to state the main result of this section.
\begin{theorem}\label{thm:pbisim}
  Two states are probabilistic delay bisimilar if and only if they are $\mcal G\circ \tspaths$-behaviourally equivalent.
\end{theorem}
\begin{proof}
  \fbox{$\Rightarrow$} Let $R$ be a probabilistic delay bisimulation, let $(X,\alpha)$ be the $\G\circ\tspaths$-coalgebra induced by $(X,\tact,P)$, and let $X\rTo^f X / R$ be the quotient map. We will construct a coalgebra on the quotient set $X/R$ in two stages. First, we construct a measure $\nu_{f(x)}$ (for each $x\in X$) on the space $\spaths{X/R}$ using the extension result (cf. Corollary~\ref{cor:meas-gen}) such that it coincides with the pushforward measure $(\alpha(x))_*$ on the open subsets $V\in\mcal O_{\spaths{X/R}}$. Second, we invoke the well-known application (taken from \cite[Proposition~2.10]{Panangaden:2009:LMP}) of Dynkin's $\lambda-\pi$ theorem to conclude that the measure $\nu_{f(x)}=(\alpha(x))_*$.

Let $x\in X$. Define a function $\mcal O_{\spaths{X/R}} \rTo^{\nu_{f(x)}} \nnreals$ as follows:
  \begin{equation}\label{eq:qconst}
    \nu_{f(x)}(V)=\tilde \mu_P(\inv\pi\inv{\spaths f} V \cap \future{\eseq x}),\qquad \text{for every}\ V\in\mcal O_{\spaths {X/R}}.
  \end{equation}
  We need a technical result proven in \cite[Lemma~24]{sokolova:2009:sacs} in order to show that $\nu_{f(x)}$ is well defined, i.e., for any $xRx'$ and open set $V\in\mcal O_{\spaths{X/R}}$ we have $\nu_{f(x)}(V)=\nu_{f(x')}(V)$. (See Theorem~\ref{thm:apbisim} for the proof of this claim).
  Moreover, the function $\nu_{f(x)}$ is a valuation, which immediately follows from \eqref{eq:qconst} and the fact that $\mu$ is a valuation. Therefore, from Corollary~\ref{cor:meas-gen}, the valuation $\nu_{f(x)}$ extends to a Borel measure $\tilde\nu_{f(x)}$ on the space $\spaths{X/R}$.

  Recall that the pushforward measure $ (\alpha(x))_*(V)=\alpha(x)(\inv{\spaths f}(V))$ (for each Borel set $V\in\Sigma_{\spaths{X/R}}$),
  is also a measure on the space $\spaths{X/R}$. Clearly, for any $V\in\mcal O_{\spaths{X/R}}$, we have $(\alpha(x))_*(V)=\nu_{f(x)}(V)=\tilde\nu_{f(x)}(V)$ due to Equation~\eqref{eq:measure-alpha}. Since both $\tilde\nu_{f(x)}$ and $(\alpha(x))_*$ are probability measures, so from \cite[Proposition~2.10]{Panangaden:2009:LMP} we get $\tilde\nu_{f(x)}=(\alpha(x))_*$. Now letting $\beta(f(x))=\tilde\nu_{f(x)}$, we find that $f$ is a coalgebra homomorphism because $\beta(f(x))(V)=\tilde\nu_{f(x)}(V)=(\alpha(x))_*(V)=\alpha(x)(\inv{\spaths{f}} V)$, for every $V\in\Sigma_{\spaths{X/R}}$.

  \fbox{$\Leftarrow$} Let $(X,\tact,P)$ be a fully probabilistic system and $(X,\alpha)$ be the corresponding $\mcal G\circ \tspaths$-coalgebra. Moreover, let $(Y,\beta)$ be a $\mcal G\circ \tspaths$-coalgebra and $X \rTo^f Y$ be a $\mcal G\circ \tspaths$ coalgebra homomorphism.
  In the appendix (Theorem~\ref{thm:apbisim}), we show that the equivalence relation $x R x' \iff f(x)=f(x')$ is a probabilistic delay bisimulation. Below we rather illustrate why $R$ is a witnessing bisimulation relation.
\end{proof}
Consider the two probabilistic transition systems drawn below,

\begin{center}
\scalebox{0.8}{
\begin{tikzpicture}
\node[draw=black] (A){
  \begin{tikzpicture}
    \node (y1) {$x_1$};
    \node (a) at ($(x1.center)+(-1.2,0.2)$) {$\alpha:$};
    \node (y2) at ($(y1.center)+(-0.5,-1)$) {$x_2$};
    \node (y3) at ($(y1.center)+(0.5,-1)$) {$x_3$};
    \node (y4) at ($(y2.center)+(0,-1)$) {$x_4$};
    \node (y1p) at ($(y1.center)+(2.5,0)$) {$x_1'$};
    \node (y2p) at ($(y1p.center)+(-0.5,-1)$) {$x_2'$};
    \node (y3p) at ($(y1p.center)+(0.5,-1)$) {$x_3'$};
    \node (y4p) at ($(y2p.center)+(0,-1)$) {$x_4'$};
    \path[->]
        (y1) edge[loop above] node[right] {$\tau,\frac 1 3$} (y1)
        (y1) edge node[above left, near end]{$\tau,\frac 1 3$} (y2)
        (y1) edge node[above right,near end]{$a,\frac 1 3$} (y3)
        (y2) edge node[left]{$b,1$} (y4)
        (y1p) edge node[above left, near end]{$\tau,\frac 1 2$} (y2p)
        (y1p) edge node[above right,near end]{$a,\frac 1 2$} (y3p)
        (y2p) edge node[left]{$b,1$} (y4p);
  \end{tikzpicture}};
\node[draw=black] (B) at ($(A.center)+(6,0)$) {
  \begin{tikzpicture}
    \node (x1) {$y_1$};
    \node (a) at ($(x1.center)+(-1.2,0.2)$) {$\beta:$};
    \node (x2) at ($(x1.center)+(-0.5,-1)$) {$y_2$};
    \node (x3) at ($(x1.center)+(0.5,-1)$) {$y_3$};
    \node (x4) at ($(x2.center)+(0,-1)$) {$y_4$};
    \path[->]
        (x1) edge node[above left, near end]{$\tau,\frac 1 2$} (x2)
        (x1) edge node[above right, near end]{$a,\frac 1 2$} (x3)
        (x2) edge node[left]{$b,1$} (x4);
  \end{tikzpicture}};
\end{tikzpicture}}
\end{center}
together with the path-based coalgebras $X\rTo^\alpha \G\spaths X$ and $Y\rTo^\beta \G\spaths{Y}$ where $X=\{x_i,x_j'\mid i,j\in\{1,2,3,4\}\}$ and $Y=\{y_i\mid i\in\{1,2,3,4\}\}$. Furthermore, let $X \rTo^f Y$ be a function defined as $f(z_i)=y_i$, where $z\in \{x,x'\}$ and $i\in\{1,2,3,4\}$. To see why the relation $R$ (as defined above) is a witnessing bisimulation, consider the equation
\begin{equation}\label{eq:illustrate}
\alpha(x_1)\Big(\bigcup_{p\in x_1 \steps{\tau^*b} [x_4]_R} \future{\qset{[p]}}\Big) = \alpha(x_1')\Big(\bigcup_{p\in x_1' \steps{\tau^*b} [x_4]_R} \future{\qset{[p]}}\Big),
\end{equation}
which can be derived from the facts $x_1 R x_1'$ and $\beta \circ f = \G\spaths f \circ \alpha$ (see Theorem~\ref{thm:apbisim} for the proof in the general case). The two terms in Equation~\ref{eq:illustrate} denote the probabilities of reaching the equivalence class $[x_4]_R$ from the states $x_1$ and $x_1'$. This can be seen, for instance, by deriving $\alpha(x_1)(\bigcup_{p\in x_1 \steps{\tau^*b} [x_4]_R} \future{\qset{[p]}})= P(x_1,\tau^* b,[x_4]_R)$ using Equation~\ref{eq:measure-alpha}, definition of $\tilde\mu_P$, Proposition~\ref{prop:separated:future}, and Lemma~\ref{lemma:step-sepclosure}. Moreover, the probability to reach the equivalence class $[x_4]_R$ from $x_1,x_1'$ is $\frac 1 2$ because $P(x_1,\tau^*b,[x_4]_R)=\sum_{i=1}^\infty (\frac 1 3)^i = \frac {1} {2} = P(x_1',\tau^*b,[x_4]_R)$.

%
\section{Discussion and conclusion}\label{sec:conc}
The main message of this paper is that behavioural equivalence in a path-based coalgebra is sufficient to capture branching bisimulation. In particular, we considered coalgebras of type $F\circ\tspaths$ over a concrete category $\mathbf C$, where $F$ is an endofunctor modelling the branching type of the system under investigation. We showed that behavioural equivalence when $F=\mcal P$ and $\mathbf C=\set$ coincides with the traditional branching bisimulation \cite{bbisim:1996}. In a similar spirit, we also showed that  behavioural equivalence when $F=\mcal G$ and $\mathbf C=\meas$ coincides with the probabilistic delay bisimulation \cite{BaierHermanns:1997,sokolova:2009:sacs,Sokolova:2005:actiontype}.

Interestingly, in the case of labelled transition systems, we can use the final chain based algorithm presented in \cite{ABHKMS12} to minimise the system with respect to branching bisimulation. The following prerequisites for this algorithm are satisfied in this context: first, a terminal object exists in $\set$; second, $\set$ is equipped with a (epi,mono)-factorisation structure; third, the functor $\mcal P\circ \tspaths$ preserves monomorphisms. However, for the probabilistic case, more research is required to find out whether the above conditions are valid or not.

In retrospect, our paper comes short in one regard when comparing with the recent works \cite{Brengos:2014,Brengos:2015:journal,Goncharov:2014:wbisim} on capturing weak bisimulation; namely, there is no abstract construction given to construct our path-based coalgebras from the system under study. In particular, we would like to construct a path-based coalgebra, for instance, $X \rTo^{\alpha'} F \spaths X$ from a given coalgebra of type $X \rTo^\alpha F(\tact \times X)$.

In this regard, it might be interesting to extend the initial work of \citeauthor{JacSok:2009:coalg:exec} \cite{JacSok:2009:coalg:exec}: Given a system $X \rTo TFX$ over $\set$ ($T$ is a monad modelling the branching type and $F$ is an endofunctor modelling the transition type), then the traces (executions) can be described as an arrow $X \rTo I$ ($X \times I \rTo I$) in the Kleisli category of $T$ with $I$ being the initial algebra of $F$. Note that this insight of \cite{JacSok:2009:coalg:exec} works under some technical requirements and it is unclear whether these requirements hold in a more general setting of $\meas$. This was already voiced by \citeauthor{henning:2013} \cite{henning:2013} in conjunction with generic trace semantics for probabilistic systems that were modelled over the base category $\meas$.

Another way to generalise the result of this paper is to consider the executions of a system as first-class citizens from the onset. Such a venture is carried out by \citeauthor{Cuijpers:2013:DCM} \cite{Cuijpers:2013:DCM} under the banner of prefix orders. Prefix orders are partially ordered sets whose principal ideals are totally ordered sets. The homomorphisms on such structures are called history preserving functions, those order preserving functions that preserve the principal ideals of the underlying ordered sets. \citeauthor{beoCuij2015:open} \cite{beoCuij2015:open} extended the theory of open maps \cite{Winskel93:open} to the concrete category setting to get a characterisation of traditional branching bisimulation. Therefore, it will be worthwhile to study whether the measure theoretic concepts proposed here can be lifted to the more general setting of prefix orders to capture probabilistic branching bisimulation. Lastly, it will be interesting to construe a notion of behavioural equivalence in the open map approach akin to the theory of coalgebras, where the notion of bisimulation is parametric to a functor modelling the branching type of system under study.

\subsection*{Acknowledgements} The authors thank the anonymous reviewers of the conference CALCO'17 for their feedbacks on an earlier draft of this paper. The authors also thank Barbara K\"{o}nig for various discussions regarding this work and for her earlier comments which led to significant improvements of this paper. The authors would also like to thank Pieter Cuijpers for his various feedbacks and continuing support over the course of this work.

This work has been carried out as part of the ``Behavioural Equivalences: Environmental Aspects,
Metrics and Generic Algorithms'' (BEMEGA) project supported by Deutsche Forschungsgemeinschaft (DFG).

\printbibliography[]

\newpage
\appendix
\renewcommand{\thesection}{\Roman{section}}
\setcounter{section}{0}
\numberwithin{theorem}{section}

\section{Nondeterministic Automata: The spectrum of bisimulations in the presence of silent transitions}\label{sec:weak-eta-del}
We have already seen how to capture branching bisimulation using the $\mathit{Path}_\sim$ functor. By adapting the modelling technique for the coalgebra $\alpha$ to saturate leading or trailing sequences of $\tau$ transitions, we can obtain delay, $\eta$ and weak bisimulation as well. For this purpose, we model $(X,\rightarrow)$ coalgebraically as $\alpha: X\rightarrow\mathcal P(\mathit{Path}_\sim(X))$ as follows:
\begin{description}
	\item[Weak Bisimulation]\index{Weak Bisimulation (Coalgebra)} \begin{equation*}\begin{aligned}\alpha(x)=\{[(p(\epsilon),a,p(\tau^na\tau^m))]\mid&\mathsf{dom}(p)=\tau^na\tau^m,\\&\forall i<n: p(\tau^i)\xrightarrow{\tau} p(\tau^{i+1}),\\&\ p(\tau^n)\xrightarrow{a} p(\tau^na), \\&\forall j<m: p(\tau^na\tau^j)\xrightarrow{\tau}p(\tau^na\tau^{j+1}), \\&a\in A_\tau\}\cup\{[x]\mid x\in X\}\end{aligned}\end{equation*}
	\item[$\eta$ Bisimulation]\index{$\eta$ Bisimulation (Coalgebra)} \begin{equation*}\begin{aligned}\alpha(x)=\{[q]\mid&q=(p(\epsilon),\tau,p(\tau),...,p(\tau^n)a,p(\tau^na\tau^m)),\\&\mathsf{dom}(p)=\tau^na\tau^m,\\&\forall i<n: p(\tau^i)\xrightarrow{\tau} p(\tau^{i+1}),\\&\ p(\tau^n)\xrightarrow{a} p(\tau^na), \\&\forall j<m: p(\tau^na\tau^j)\xrightarrow{\tau}p(\tau^na\tau^{j+1}), \\&a\in A_\tau\}\cup\{[x]\mid x\in X\}\end{aligned}\end{equation*}
	\item[Delay Bisimulation]\index{Delay Bisimulation} \begin{equation*}\begin{aligned}\alpha(x)=\{[q]\mid&q=(p(\epsilon),a,p(\tau^na),\tau,p(\tau^na\tau),...,p(\tau^na\tau^m)),\\&\mathsf{dom}(p)=\tau^na\tau^m,\\&\forall i<n: p(\tau^i)\xrightarrow{\tau} p(\tau^{i+1}),\\&\ p(\tau^n)\xrightarrow{a} p(\tau^na), \\&\forall j<m: p(\tau^na\tau^j)\xrightarrow{\tau}p(\tau^na\tau^{j+1}), \\&a\in A_\tau\}\cup\{[x]\mid x\in X\}\end{aligned}\end{equation*}
\end{description}
Using an example that allows to distinguish all four notions of bisimulation, we will demonstrate these ways of modelling LTS.

\newpage
\begin{example}
	Consider the set of actions $A_\tau=\{a,b,\tau\}$ that contains two non-silent transition $a$ and $b$. We define a transition system over $A_\tau$ and the state set $X=\{A,B,C,D,E\}$ as depicted below:
	
\begin{center}
	\begin{tikzpicture}[x=2.5cm,y=1.5cm,double distance=2pt]
      \node (a) at (1,2) {$A$} ;
      \node (b) at (1,0) {$B$} ;
      \node (c) at (1,4) {$C$} ;
      \node (d) at (0,1) {$D$} ;
      \node (e) at (2,1) {$E$} ;
      \begin{scope}[->]
        \path[shorten <=1pt]
          (a) edge node[above] {$a,\tau$} (d)
          		edge node[above] {$\tau$} (e)
          		edge[loop above] node {$b$} (a);
        \path
          (b) edge node[above] {$\tau$} (d)
          		edge[loop below] node {$b$} (b)
          edge node[above] {$\tau$} (e) ;
        \path
          (c) edge node[left] {$a,b,\tau$} (d)
          		edge node[right] {$\tau$} (e)
          		edge[loop above] node {$b$} (c);
        \path
          (e) edge node[above] {$a$} (d);
      \end{scope}
    \end{tikzpicture}
\end{center}
We will now model this system using the $\mathcal P\circ\mathsf{Path}_\sim$ functor in order to get each of the four notions of bisimilarity as the behavioural equivalence and compute Algorithm B (or, equivalently, Algorithm C). We will identify each equivalence class of paths with its paths and write down the computed sets of paths explicitly. The paths that are redundant are printed in grey. Redundant in this context means that there exists another representative of the same equivalence class in the set.

\begin{description}
	\item[Branching Bisimulation] To obtain branching bisimulation as behavioural equivalence, we model the transition system by the following table:
\begin{center}
	\scalebox{0.9}{\begin{tabular}{|c|c|}\hline
		\multicolumn{2}{|c|}{$\alpha\colon X\rightarrow\mathcal P(\mathit{Path}_\sim(X))$}\\\hline
		$A$&$\{(A),(A,b,A),(A,b,A,\tau,D),(A,a,D),(A,\tau,D),(A,\tau,E),(A,\tau,E,a,D),(A,b,A,\tau,E)\}$\\
		$B$&$\{(B),(B,b,B),(B,b,B,\tau,E),(B,\tau,E),(B,\tau,E,a,D),(B,\tau,D),(B,b,B,\tau,D)\}$\\
		$C$&$\{(C),(C,b,C),(C,b,C,\tau,D),(C,a,D),(C,b,D),(C,\tau,D),(C,\tau,E),(C,\tau,E,a,D)\}$\\
		$D$&$\{(D)\}$\\
		$E$&$\{(E),(E,a,D)\}$\\\hline
	\end{tabular}}
\end{center}

	\item[Weak Bisimulation] To obtain weak bisimulation as behavioural equivalence, we model the transition system by the following table:
\begin{center}
	\scalebox{.9}{\begin{tabular}{|c|c|}\hline
		\multicolumn{2}{|c|}{$\alpha\colon X\rightarrow\mathcal P(\mathit{Path}_\sim(X))$}\\\hline
		$A$&$\{(A),(A,b,A),(A,b,D),(A,b,E),(A,\tau,D),(A,\tau,E),(A,a,D)\}$\\
		$B$&$\{(B),(B,b,B),(B,b,D),(B,b,E),(B,\tau,D),(B,\tau,E),(B,a,D)\}$\\
		$C$&$\{(C),(C,b,C),(C,b,D),(C,b,E),(C,\tau,D),(C,\tau,E),(C,a,D)\}$\\
		$D$&$\{(D)\}$\\
		$E$&$\{(E),(E,a,D)\}$\\\hline
	\end{tabular}}
\end{center}
	
		\item[$\eta$ Bisimulation] To obtain $\eta$ bisimulation as behavioural equivalence, we model the transition system by the following table:
\begin{center}
	\scalebox{.9}{\begin{tabular}{|c|c|}\hline
		\multicolumn{2}{|c|}{$\alpha\colon X\rightarrow\mathcal P(\mathit{Path}_\sim(X))$}\\\hline
		$A$&$\{(A),(A,\tau,D),(A,\tau,E),(A,\tau,E,a,D),(A,b,A),(A,b,D),(A,b,E),(A,a,D)\}$\\
		$B$&$\{(B),(B,\tau,D),(B,\tau,E),(B,\tau,E,a,D),(B,b,B),(B,b,D),(B,b,E)\}$\\
		$C$&$\{(C),(C,\tau,D),(C,\tau,E),(C,\tau,E,a,D),(C,b,C),(C,b,D),(C,b,E),(C,a,D)\}$\\
		$D$&$\{(D)\}$\\
		$E$&$\{(E),(E,a,D)\}$\\\hline
	\end{tabular}}
\end{center}

		\item[Delay Bisimulation] Finally, to obtain delay bisimulation as behavioural equivalence, we model the transition system by the following table:
\begin{center}
	\scalebox{.9}{\begin{tabular}{|c|c|}\hline
		\multicolumn{2}{|c|}{$\alpha\colon X\rightarrow\mathcal P(\mathit{Path}_\sim(X))$}\\\hline
		$A$&$\{(A),(A,\tau,D),(A,\tau,E),(A,a,D),(A,b,A),(A,b,A,\tau,D),(A,b,A,\tau,E)\}$\\
		$B$&$\{(B),(B,\tau,D),(B,\tau,E),(B,a,D),(B,b,B),(B,b,B,\tau,D),(B,b,B,\tau,E)\}$\\
		$C$&$\{(C),(C,\tau,D),(C,\tau,E),(C,a,D),(C,b,C),(C,b,C,\tau,D),(C,b,C,\tau,E),(C,b,D)\}$\\
		$D$&$\{(D)\}$\\
		$E$&$\{(E),(E,a,D)\}$\\\hline
	\end{tabular}}
\end{center}
\end{description}
\end{example}

\section{Proofs concerning Section~\ref{sec:prelim}}\label{app-prelim}
\begin{proposition}[Proposition~\ref{prop:trace-last-pres}]
  Let $X \rTo^f Y$ be a function. Then, for every path $p\in \paths X$ we have $\trace p = \trace {fp}$ and $f \last p = \last{fp}$.
\end{proposition}
\begin{proof}
  Let $p\in\paths X$. Then, $\trace p = \max \dom p = \max \dom {fp} =\trace {fp}$. Furthermore, $\last{fp} = fp \max \dom{fp}= fp \max \dom p = f \last p$.
\end{proof}

\begin{theorem}[Theorem~\ref{thm:stutter-unique}]
For any path there is a unique stutter basis for the given path.
\end{theorem}
\begin{proof}
  Let $\phi,\phi'$ be the two witnessing stutter bases for $p$ with $\sigma=\max \dom p$. We prove $\phi=\phi'$ by structural induction on the words in $\history\sigma$. Without loss of generality, let $\sigma'a\in \history \sigma$. Then, we distinguish the following cases:
  \begin{enumerate}
    \item Let $a\in\act$. Then, by the induction hypothesis we have $
        \phi(\sigma')=\phi'(\sigma') \implies
        \phi(\sigma')a=\phi'(\sigma')a \implies
        \phi(\sigma'a)=\phi'(\sigma'a).
        $
    \item Let $a=\tau$. Let $p(\sigma'\tau)=p(\sigma')$. (The proof of the other case when $p(\sigma'\tau)\neq p(\sigma')$ is similar to the previous case.) Again using the induction hypothesis we get
        $
        \phi(\sigma')=\phi'(\sigma') \implies \phi(\sigma'\tau)=\phi'(\sigma'\tau)
        $.\qedhere
  \end{enumerate}
\end{proof}

\begin{lemma}[Lemma~\ref{lemma:stutt-hpres}]
  Let $\phi$ be the stutter basis for a given path $p$ with $\dom p=\history {\sigma}$. Then, $\phi(\history {\sigma})=\ \history{\phi(\sigma)}.$
\end{lemma}
\begin{proof}
  Let $\phi$ be the stutter basis for the path $p$ with $\dom p=\history {\bar\sigma}$. To prove that $\phi$ is history preserving, it suffices to show that $\phi$ is order preserving and $\phi$ satisfies the backward simulation property (i.e., $\forall_{\sigma',\sigma}\ \sigma'\preceq \phi(\sigma) \implies \exists_{\sigma''}\ \sigma''\preceq \sigma\land \phi(\sigma'')=\sigma'$). For order preservation, consider $\sigma \preceq \sigma'$, for some $\sigma,\sigma'\in\history{\bar \sigma}$. Without loss of generality, let $\sigma'=\sigma a$. Then, either $\phi(\sigma') = \phi (\sigma) a$ or $\phi (\sigma') = \phi (\sigma)$. I.e., in either case we have $\phi (\sigma) \preceq \phi(\sigma')$.

  For the backward simulation, let $\sigma' \preceq \phi(\sigma)$. Again, without loss of generality, let $\phi(\sigma)=\sigma' a$. If $a\neq\eseq{}$, then by the construction of $\phi$, we know that there is some $\sigma''$ such that $\sigma=\sigma''a \land \phi(\sigma'')=\sigma'$. If $a=\eseq{}$, then again by the construction of $\phi$ we know that there is some $\sigma''$ such that $\sigma=\sigma''\tau\land \phi(\sigma'')=\sigma'$. Thus, there is some $\sigma''$ such that $\sigma''\preceq \sigma$ and $\phi(\sigma'')=\sigma'$.
\end{proof}

\begin{lemma}[Lemma~\ref{lemma:stutt-preserve}]
  For any $p\in\paths X$ and any $X\rTo^f Y$, we have $f\circ \stutter p \sim f\circ p$.
\end{lemma}
\begin{proof}
  Let $\phi,\phi',\phi''$ be the stutter bases for the paths $fp,p,f\stutter p$, respectively. Consider the following commutative diagram with $\dom p = \history \sigma$.
  \begin{diagram}[height=2.5em]
    \phi'\history\sigma & \rTo^{\phi''} & & \phi''\phi'\history\sigma\\
    \uTo^{\phi'} & \rdTo^{\stutter p} &\ldTo(1,4)_{\stutter{(f\stutter p)}} & \\
    \history \sigma & \rTo^p & X & \\
    \dTo^{\phi} & & \dTo^f &\\
    \phi\history \sigma & \rTo^{\stutter{(fp)}} & Y &
  \end{diagram}
We prove that $\phi=\phi''\phi'$ by induction on the words in $\history \sigma$. Without loss of generality, let $\sigma' a\in\history\sigma$. If $a\in A$, then the result follows directly from the inductive hypothesis. So, let $a=\tau$. Then we distinguish the following cases.
    \begin{enumerate}
      \item Let $fp(\sigma'\tau)=fp(\sigma')$ and let $p(\sigma'\tau)=p(\sigma')$. Then we have $\phi'(\sigma'\tau)=\phi'(\sigma') \implies \phi''\phi'(\sigma'\tau)=\phi''\phi'(\sigma')$. And using induction hypothesis we get $\phi''\phi'(\sigma'\tau)=\phi(\sigma')$. Moreover, $\phi(\sigma')=\phi(\sigma'\tau)$ since $fp(\sigma'\tau)=fp(\sigma')$.
      \item Let $fp(\sigma'\tau)=fp(\sigma')$ and let $p(\sigma'\tau)\neq p(\sigma')$. Then we have $\phi'(\sigma'\tau)=\phi'(\sigma')\tau$. Furthermore, $fp(\sigma'\tau)=fp(\sigma')\implies f\stutter p \phi'(\sigma'\tau) = f\stutter p \phi'(\sigma')$. I.e., $f\stutter p \phi'(\sigma')\tau=f\stutter p \phi'(\sigma')$. So by the construction of a stutter basis we have $\phi''(\phi'(\sigma')\tau)=\phi''\phi'(\sigma')$. Using induction hypothesis we get $\phi''(\phi'(\sigma')\tau) = \phi''\phi'(\sigma') = \phi (\sigma')$. Moreover, $\phi(\sigma')=\phi(\sigma'\tau)$ since $fp(\sigma'\tau)=fp(\sigma')$. Thus, $\phi''\phi'(\sigma'\tau) = \phi''(\phi'(\sigma')\tau) =  \phi(\sigma'\tau)$.
      \item Let $fp(\sigma'\tau)\neq fp(\sigma')$ and let $p(\sigma'\tau)=p(\sigma')$. This case is not possible.
      \item Let $fp(\sigma'\tau)\neq fp(\sigma')$ and let $p(\sigma'\tau)\neq p(\sigma')$. Then we have $\phi(\sigma'\tau)=\phi(\sigma')\tau$ and $\phi'(\sigma'\tau)=\phi'(\sigma')\tau$. Using $p=\stutter p \phi'$ in $fp(\sigma'\tau)\neq fp(\sigma')$ we get $f\stutter p \phi'(\sigma'\tau) \neq f\stutter p \phi'(\sigma')$. Then, by the construction of a stutter basis we get $\phi''\phi'(\sigma'\tau)=\phi''\phi'(\sigma')\tau$.
          I.e., using the induction hypothesis we get $\phi''\phi'(\sigma'\tau)=\phi(\sigma')\tau = \phi(\sigma'\tau)$.
    \end{enumerate}
Thus, $\phi=\phi''\phi'$. I.e., $\stutter{(fp)}\phi=fp=f\stutter p\phi'=\stutter{(f\stutter p)} \phi'' \phi'=\stutter{(f\stutter p)} \phi$. Moreover, the function $\phi$ is an epi; therefore, we conclude that $\stutter{(fp)}=\stutter{(f\stutter p)}$.
\end{proof}

\begin{theorem}[Theorem~\ref{thm:spath-functor}]
  The mapping $\spaths f$ given in \eqref{eq:spaths} is well defined and is an endofunctor on the category of sets.
\end{theorem}
\begin{proof}
  We first show that the definition given in \eqref{eq:spaths} is independent of the chosen representative. I.e., If $[p]_\sim=[p']_\sim$ then $\spaths f [p]_\sim = \spaths f[p']_\sim$. Let $p\sim p'$. Then,
  $\stutter p =\stutter {p'} \implies f\circ \stutter p = f\circ \stutter{p'} \implies \stutter{(f\circ \stutter p)} =\stutter{(f\circ \stutter {p'})}.$
  Furthermore, using Lemma~\ref{lemma:stutt-preserve} we get $\stutter{(f\circ p)}=\stutter{(f\circ p')}$. Thus, $f\circ p \sim f \circ p'$.

  The mapping $\tspaths$ preserves identity:
  \[\spaths {\id X}[p]_\sim=[\id{ } \circ p]_\sim=[p]_\sim =\id{\spaths {X}}[p]_\sim.\]
  The mapping $\tspaths$ preserves composition. Let $X \rTo^f Y$ and $Y \rTo^g Z$. Then,
  \begin{align*}
    \spaths{g}\circ \spaths f[p]_\sim =&\ \spaths{g} ([f \circ p]_\sim)\\
    =&\ [g\circ f \circ p]_\sim\\
    =&\ \spaths{g\circ f}[p]_\sim. &&\qedhere
  \end{align*}
\end{proof}

\begin{lemma}[Lemma~\ref{blemma}]
  Let $p\in\paths X,q\in \paths Y$, and $X \rTo^f Y$.
  \begin{enumerate}
    \item If $fp\sim q$ and $\trace q \in \tau^*a\tau^*$ then $q(\eseq{})= fp(\eseq{})\land \trace p \in \tau^*a \tau^*$.
    \item $\last p = \last {\stutter p}$.
    \item If $fp \sim q$, then $f(\last p)=\last q$.
    \item $\pi_Y \circ \paths f = \spaths f \circ \pi_X$.
    \item Let $V_i\subseteq \spaths Y$ (for $i\in I$) be a family of pairwise disjoint sets. Then,
  \[\inv{\spaths f} \Big(\bigcup_{i\in I} V_i\Big) \ = \ \bigcup_{i\in I} \inv {\spaths f} (V_i)\enspace.\]
  \end{enumerate}
\end{lemma}
\begin{proof}
  \begin{enumerate}
    \item For any paths $p\in\paths X,q\in \paths Y$ with $X \rTo^f Y$ and $fp\sim q$, the following properties hold.
          \begin{enumerate}
            \item $q(\eseq{})= fp(\eseq{})$.
            \item If $\trace q \in \tau^*a\tau^*$ then $\trace p \in \tau^*a \tau^*$.
          \end{enumerate}
        For (a), we know from the the construction of the stutter invariant to $q$: $q(\eseq{})=\stutter q(\eseq{})$. Thus, we have $q(\eseq{})=\stutter q(\eseq{})=\stutter{(fp)}(\eseq{})=fp(\eseq{})$.

        For (b), let $\trace q\in\tau^*a \tau^*$. Then, it suffices to show that $\trace {\stutter q}\in\tau^*a\tau^*$. So let $\dom q=\history \tau^n a\tau^m$, for some $n,m\geq 0$ and let $\phi$ be the stutter basis for $q$. Then, by the construction of $\phi$ we have $\phi(\tau^n a\tau^m)=\tau^{n'} a\tau^{m'}$ with $n'\leq n$ and $m'\leq m$. Thus, $\trace {\stutter q}\in\tau^*a\tau^*$. Therefore, we have $\trace {\stutter{(fp)}} = \trace {\stutter q} \in \tau^*a\tau^*$. And since $\trace p = \trace {(fp)}$ (Proposition~\ref{prop:trace-last-pres}) we conclude that $\trace{fp}\in\tau^*a\tau^* \implies \trace p \in \tau^*a\tau^*$.
    \item Let $\dom p= \history \sigma$. Then, we derive
          \begin{align*}
            &\ p(\sigma) = \stutter p \phi(\sigma) && (\because p=\stutter p \circ \phi )\\
            \iff &\ p \max \history \sigma = \stutter p \max \history {(\phi \sigma)} && (\because \sigma = \max \history \sigma )\\
            \iff &\ \last p = \stutter p \max \phi (\history \sigma)  && (\because \history {\phi(\sigma)}=\phi (\history\sigma))\\
            \iff &\ \last p = \last {\stutter p}. 
          \end{align*}
    \item Let $p\in\paths X$ and $q\in\paths Y$ with $\stutter{(fp)} =\stutter q$. Then, we have
          \[\last q = \last {\stutter q }= \last{\stutter{(fp)}} = \last {fp} = f \last p.  \]
    \item Trivial.
    \item Let $\qset{[p]}\in\inv {\spaths f} (\bigcup_{i\in I} V_i)$. I.e., there is a unique $j\in I$ (because of disjointness) with $\qset{[fp]}\in V_i$. Thus, $\qset{[p]}\in\bigcup_{i\in I}\inv{\spaths f}( V_i)$. The other direction, i.e., the case $\bigcup_{i\in I} \inv {\spaths f} (V_i) \subseteq \inv{\spaths f} \Big(\bigcup_{i\in I} V_i\Big) $ is trivial.\qedhere
  \end{enumerate}
\end{proof}

\section{Proof concerning Section~\ref{sec:lts}}

\begin{theorem}[Theorem~\ref{thm:bisim}]\label{athm:bisim}
  Let $(X,\tact,\rightarrow)$ be a labelled transition system and $(X,\pi_X\circ \alpha)$ be the corresponding $\mcal P\circ \tspaths$-coalgebra. Then, two states $x,x'\in X$ are branching bisimilar if and only if the states $x,x'$ are $\mcal P\circ \tspaths$-behaviourally equivalent.
\end{theorem}
\begin{proof}
Let $(Y,\beta)$ be a $\mcal P\circ \tspaths$-coalgebra and let $X \rTo^f Y$ be a $\mcal P\circ \tspaths$-coalgebra homomorphism. Define a relation $x R x' \iff f(x)=f(x')$.
  Let $x_1 \step a x_2$ and $x_1 R x_1'$. Below, we only consider $a\in\act$. Define a path $\{\eseq{},a\} \rTo^p X$ as $p(\eseq{})=x_1$ and $p(a)=x_2$. Then, by the construction of $\alpha$ we have $p\in \alpha(x)$. Furthermore,
  {\allowdisplaybreaks
  \begin{align*}
    &\ \qset{[fp]} \in  \mcal P(\spaths f)\circ \pi_X\circ\alpha(x)\\
    &\ \implies  \qset{[fp]}\in \beta(f(x)) && (\because \mcal P( \spaths f)\circ \pi_X \circ \alpha=\beta\circ f)\\
    &\ \implies \qset{[fp]}\in \beta(f(x')) && (\because f(x)=f(x'))\\
    &\ \implies \qset{[fp]}\in \mcal P(\spaths f)\circ\pi_X\circ \alpha(x'). && (\because \mcal P( \spaths f)\circ \pi_X \circ \alpha=\beta\circ f)
  \end{align*}
  }
  I.e., there is a path $p'\in\exec{x_1'}$ such that $fp\sim fp'$.
  Note that $\stutter{(fp)}=fp$ since $p=\stutter p$ and $p$ is a one-step execution with observation $a\neq\tau$. Let $\phi$ be the stutter basis for $p'$. Next, we claim that
  \begin{equation}\label{eq-2}
    fp'(\sigma) = fp\phi(\sigma), \qquad \text{for any}\ \sigma\in\dom {p'}.
  \end{equation}
  To see this, let $\sigma \in\dom {p'}$. Then, $fp'(\sigma)=\stutter{(fp')}\phi(\sigma)=\stutter{(fp)}\phi(\sigma)=
  fp\phi(\sigma)$. Notice that $\stutter{(fp')}=\stutter{(fp)}$ because $fp\sim fp'$. This proves the above claim.

  Using Lemma~\ref{blemma:stutt-pres-laststate} and construction of $\alpha$ we get $\trace {p'}\in\tau^* a$ since $\trace p=a$. So let $\dom {p'}=\tau^m a$, for some $m\geq 0$. Next, we claim that for any $\sigma\in\history \tau^m$ we have $fp'(\sigma)=f(x)$. So, let $\sigma\in\history \tau^m$. Then, using \eqref{eq-2} we get $fp'(\sigma)= fp\phi(\sigma)$. Clearly, by the construction of $\phi$ we have $\phi(\sigma)\in\tau^{m'}$ (for some $m'\leq m$). But, $\dom p=\{\eseq{},a\}$, so $\phi(\sigma)=\eseq{}$. Thus, $fp'(\sigma)=fp(\eseq{})=f(x)$. I.e.,
  $x_1' \steps{\eseq{}} p'(\tau^m)$ and $fp'(\tau^m)=f(x_1).$
  Moreover, using \eqref{eq-2} we also have
  $
  fp'(\tau^m a)=fp\phi(\tau^m a)=fp\phi(\tau^m) a =fp(a)=f(x_2).$
  Thus, $x'\steps{\eseq{}}p'(\tau^m)\step{a}p'(\tau^m a)$ and $fp'(\tau^m a)=f(x_2)$. And construction of the relation $R$ gives $x_1' \ R \ p'(\tau^m)$ and $x_2 \ R \ p'(\tau^m a)$.
\end{proof}

\section{Proofs concerning Section~\ref{sec:pmeas-paths}}

\begin{lemma}[Lemma~\ref{lemma:stutt-hreflect}]
  Let $X$ be a set. Then we have the following property:
  \[
  \forall_{p_1,p_2\in\paths X}\ \stutter p_1 \preceq \stutter p_2 \implies \exists_{p\in\paths X}\ p \sim p_1 \land p\preceq p_2.
  \]
\end{lemma}
\begin{proof}
  Assume $\stutter{p_1}\leq \stutter{p_2}$, then by definition
$\stutter{p_2}\mid_{\dom{\stutter{p_1}}}=\stutter{p_1}$. Observe
that due to the structure of paths, it follows that there exists a word
$w\in\dom{\stutter{p_2}}$ such that
$\stutter{p_2}\mid_{\dom{\stutter{p_1}}}=\stutter{p_2}\mid_{\history\sigma}$,
in fact it is true that $\sigma=\max \dom{\stutter{p_1}}$. We will show
that then there must also be a word $\sigma'\in\dom{p_2}$ such that
$\stutter{(p_2\mid_{\history{\sigma'}})}=\stutter{p_2}\mid_{\history \sigma}$,
i.e. $\stutter{(p_2\mid_{\history{\sigma'}})}=\stutter{p_1}$, thus
$p_2\mid_{\history{\sigma'}}\sim p_1$ and, by definition,
$p_2\mid_{\history{\sigma'}}\leq p_2$ as required.

It remains to show that whenever for a given path $p$, we choose any
$\sigma\in\dom{\stutter{p}}$, then there exists a word
$\sigma'\in\dom p$, such that
$\stutter{(p\mid_{\history{\sigma'}})}=\stutter{p}\mid_{\history{\sigma}}$.
This follows from the definition of a stutter basis; just take $\sigma'$ as a
pre-image of the stutter basis $\phi$ for $p$ of $\sigma$ (this need not be
unique). The result is clear by choosing the stutter basis
$\phi\mid_{\history{\sigma'}}$ for $p\mid_{\history{\sigma'}}$.
\end{proof}

\begin{theorem}[Theorem~\ref{thm:historypres}]
  For any $X \rTo^f Y$, the function $\paths f$ is strictly order preserving. Moreover, $\paths f$ also satisfies the following property:
  \[\forall_{p\in\paths X,q\in\paths Y}\ q \prec f(p) \implies \exists_{p'\in\paths X}\ p' \prec p \land f(p')=q\enspace.\]
\end{theorem}
\begin{proof}
  Let $p,p'\in\paths X$ with $p\prec p'$. It suffices to show that $fp \prec fp'$. Clearly, $\dom {fp}=\dom p \subset \dom p' = \dom{fp'}$. Moreover, for any $\sigma\in\dom {fp}$ we have $fp(\sigma)=f(p(\sigma))=f(p'(\sigma))=fp'(\sigma)$ Thus, $fp \prec fp'$.

  For the other property, let $q \prec fp$, for some $p\in\paths X,q\in \paths Y$. Define a function $\dom q \rTo^{p'} X$ as follows: $\forall_{\sigma\in\dom q}\ p'(\sigma)=p(\sigma)$. Clearly, $\dom q$ is a prefix-closed set due to the construction of $q$. Therefore, we just need to verify that $p'\preceq p$. Clearly, $\dom {p'}=\dom {q}\subset \dom {fp}=\dom p$. Thus, by construction of $p'$ we have $p'\prec p$. Lastly, $fp'=q$ follows from $\dom {fp'}=\dom q$ and for any $\sigma\in \dom q$ we have $fp'(\sigma)=fp(\sigma)=q(\sigma)$ (since $q\prec fp$).
\end{proof}

\begin{lemma}[Lemma~\ref{lemma:prec-qset-welldef}]
  The relation $\preceq$ on the set $\spaths X$ is a well-defined partial order.
\end{lemma}
\begin{proof}
  We first prove that the relation $\preceq$ is well defined.
  Let $\qset{[p_1]}=\qset{[p_2]}$ and $\qset{[q_1]}=\qset{[q_2]}$. Then,
  \[\qset{[p_1]} \preceq \qset{[q_1]} \iff \stutter p_1 \preceq \stutter q_1 \iff \stutter p_2 \preceq \stutter q_2  \iff \qset{[p_2]} \preceq \qset{[q_2]}\enspace.\]
  Reflexivity is obvious. For antisymmetry, let $\qset{[p]} \preceq \qset{[q]}$ and $\qset{[q]}\preceq \qset{[p]}$. Then, we find $\stutter p \preceq \stutter q$ and $\stutter q \preceq \stutter p$. I.e., $\stutter p=\stutter q$. Therefore, $\qset {[p]} = \qset {[q]}$.
  For transitivity, let $\qset{[p_1]}\preceq \qset{[p_2]}$ and $\qset{[p_2]}\preceq \qset{[p_3]}$. Then, $\stutter p_1 \preceq \stutter p_2$ and $\stutter p_2 \preceq \stutter p_3$. Thus, $\stutter p_1 \preceq \stutter p_3$; hence, $\qset{[p_1]}\preceq \qset{[p_3]}$.
\end{proof}

\begin{theorem}[Theorem~\ref{thm:quotientmea}]
  The quotient function $\paths X \rTo^\pi \spaths X$ is order preserving. Consequently, the quotient function $\paths X \rTo^\pi \spaths X$ is Borel measurable, where the sigma-algebra on paths is given by $\Sigma_{\spaths X}=\mcal B(\mcal O_{\spaths X})$.
\end{theorem}
\begin{proof}
  Let $p\preceq q$ with $p,q\in\paths X$. Then, it suffices to show that $\stutter p \preceq \stutter q$. Let $\phi,\phi'$ be the stutter bases for the paths $p,q$, respectively. Without loss of generality, let $\max \dom q =\max \dom p a$, for some $a\in \tact$. Also, let $\dom p=\history\sigma$, for some $\sigma\in\tact^*$. For the function $q$, we write the restriction of $q$ to the sub-domain $\history\sigma$ as $q|_{\sigma}$. Note that $q |_{\sigma} = p$ (since $p\preceq q$). Next, we claim that the restriction $\phi'|_{\sigma}$ is a stutter basis for $p$. Let $\sigma'a\in\history\sigma$. Then we distinguish the following cases:
  \begin{itemize}
    \item Let $a=\tau$ and $p(\sigma'\tau)=p(\sigma')$. Then, $p(\sigma'\tau)=p(\sigma') \implies q(\sigma'\tau)=q(\sigma')$ (since $p\preceq q$). Thus, $\phi'(\sigma'\tau)=\phi'(\sigma')$, i.e., $\phi'|_{\sigma}(\sigma'\tau)=\phi'|_{\sigma}(\sigma')$.
    \item Let $a=\tau$ and $p(\sigma'\tau)\neq p(\sigma')$. Then, $q(\sigma'\tau)\neq q(\sigma')$. Thus, $\phi'(\sigma'\tau)=\phi'(\sigma')\tau$ which further implies that $\phi'|_{\sigma}(\sigma'\tau)=\phi'|_{\sigma}(\sigma')\tau$.
    \item Let $a\in\act$. Similar to the previous case.
  \end{itemize}
  Thus, $\phi'$ is a stutter basis for $p$ and from the uniqueness of stutter basis for $p$ (cf. Theorem~\ref{thm:stutter-unique}) we have $\phi'|_{\sigma}=\phi$. Therefore, $\phi(\sigma')=\phi'(\sigma')$, for all $\sigma'\in\history\sigma$. Furthermore, from the construction of $\phi'$ we have $\phi'(\sigma a)=\phi'(\sigma) a$ or $\phi'(\sigma \tau)=\phi'(\sigma)$. Thus, in either case, we have $\phi(\sigma)\preceq\phi'(\sigma a)$. Hence, $\dom {\stutter p} = \history {\phi(\sigma)} \subseteq \history {\phi'(\sigma a)}=\dom {\stutter q}$. Let $\sigma'\in\dom {\stutter p}$. Since $\phi$ is the stutter basis for $p$, we know that there is $\sigma''\in \dom p$ such that $\phi(\sigma'')=\sigma'$. Then,
  \begin{align*}
    \stutter q (\sigma')=&\ \stutter q \phi(\sigma'')\\
    =&\ \stutter q \phi'(\sigma'') & (\because \phi'|_{\sigma}=\phi)\\
    =&\ q (\sigma'')= p(\sigma'')=\stutter p \phi(\sigma'')=\stutter p (\sigma').
  \end{align*}
  Thus, $\stutter p \preceq \stutter q$, whenever $p\preceq q$.
\end{proof}

\begin{theorem}[Theorem~\ref{thm:spath-opres}]
  For any $X \rTo^f Y$, the function $\spaths f$ is order preserving. Thus, the function $\spaths f$ is Borel measurable.
\end{theorem}
\begin{proof}
  Let $p,q\in\paths X$ with $\qset{[p]} \preceq \qset{[q]}$. Since $\qset{[p]}\preceq \qset{[q]}$, we have $\stutter p \preceq \stutter q$. Moreover, from Theorem~\ref{thm:historypres} we know that $\paths f(\stutter p) \preceq \paths f(\stutter q)$. Thus, $f\circ(\stutter p) \preceq f\circ(\stutter q)$. And using Theorem~\ref{thm:quotientmea} we get $\stutter{(f\circ\stutter p)} \preceq \stutter{(f\circ\stutter q)}$. Lastly, Lemma~\ref{lemma:stutt-preserve} gives us $\stutter {(f\circ p)} \preceq \stutter{(f\circ q)}$; thus, $\qset{[f\circ p]} \preceq \qset{[f\circ q]}$, i.e., $\spaths f\qset{[p]}\preceq \spaths f \qset{[q]}$.
\end{proof}

\begin{proposition}[Proposition~\ref{prop:irr-downclosed}]
  A directed subset of a prefix order is always totally ordered. In addition, an irreducible downward closed subset of a prefix order is always totally ordered. 
\end{proposition}
\begin{proof}
  Let $D \subseteq \paths X$ be a directed subset of paths with $p,q\in D$ for $p\neq q$. Then, we find some $p'\in D$ such that $p\preceq p'$ and $q\preceq p'$, since $D$ is directed. And downward totality tells us $p\preceq q$ or $q\preceq p$. Thus, every directed set is totally order.

  Furthermore, it is well-known that every irreducible downward closed subset of a partial order is an ideal (see, e.g., \cite{Fink-Goub:2009:wsts}). Thus, in particular, every directed subset of paths is totally ordered.
\end{proof}

\begin{lemma}[Lemma~\ref{lemma:dtotal-lpaths}]
  The set $X^\infty$ is prefix ordered by the relation $\preceq '$, whenever $(X,\preceq)$ is a prefix order.
\end{lemma}
\begin{proof}
  That the relation $\preceq'$ is a partial order follows directly from the construction of $\preceq'$. The interesting case is to show that the $X^\infty$ satisfies the downward totality. Let $p_1,p_2,p_3\in X^\infty$ such that $p_1 \preceq' p_3$ and $p_2 \preceq' p_3$. If $p_3\in X$ then the result follows directly from $\preceq$. So suppose $p_3=\infty_C$, for some non-sober subset $C\subseteq X$. Then, we find $p_1,p_2\in C$. And since irreducible closed sets are totally ordered (Proposition~\ref{prop:irr-downclosed}), we have $p_1 \preceq p_2 \lor p_2 \preceq p_1$, i.e., $p_1 \preceq' p_2 \lor p_2 \preceq' p_1$.
\end{proof}

\begin{proposition}[Proposition~\ref{prop:scott-sober}]
  Let $(X,\preceq)$ be a simple prefix order. A subset $U\subseteq X^\infty$ is \emph{Scott} open if and only if $U$ is upward closed and it is inaccessible by directed joins, i.e., for any directed set $D\subseteq X^\infty$ if $\sup D$ exists and $\sup D\in U$ then $D \cap U \neq \emptyset$. Let $\mcal S_{X^\infty}$ denote the collection of Scott open subsets of $X^\infty$. Then, the space $(X^\infty,\mcal S_{X^\infty})$ is a sober space.
\end{proposition}
\begin{proof}
  Let $C\subseteq X^\infty$ be an irreducible closed set. Recall that $C$ is Scott closed if and only if $C$ is downward closed and for any directed set $D\subseteq C$, if $\sup D$ exists then $\sup D \in C$. Moreover, an irreducible closed set is totally ordered (cf. Proposition~\ref{prop:irr-downclosed}) and every chain is a directed set. If $C=\history p$ for some $p$, then clearly $\sup C$ exists and $C=\closure p =\history p$. On the other hand, if $C$ is an infinite chain, then $C$ can be either a non-sober set w.r.t. Alexandroff topology with or without the limit point $\infty_C$. However, if $C$ is without the limit point $\infty_C$, then $C$ is not a Scott-closed set because the supremum $\sup C=\infty_C$ exists in the $X^\infty$ and $\infty_C\not\in C$. Thus, if $C$ is an infinite chain then $C=\history \infty_C$, whenever $C$ is an irreducible Scott-closed set.
  Clearly, we have $C=\history \sup {\ C}$, for any irreducible Scott closed set $C$. Thus, the space $(X^\infty,\mcal S_{X^\infty})$ is sober.
\end{proof}

\begin{theorem}[Theorem~\ref{thm:val-scott-extension}]
  Let $(X,\preceq)$ be a prefix order and let $\mcal O_{X} \rTo^\mu \nnreals$ be a locally finite and Scott-continuous valuation. Then, the function $\mcal S_{X^\infty} \rTo^{\tilde \mu} \nnreals$ defined in the following way:
  \[\tilde \mu(V) = \mu(V \cap X) \quad \text{(for every Scott-open set $V\in \mcal S_{X^\infty}$)}\]
  is a Scott-continuous valuation. If $X$ is simple then $\tilde\mu$ is locally finite.
\end{theorem}
\begin{proof}
  First we note that $V \cap X \in \mcal O_{X}$, whenever $V\in\mcal S_{X^\infty}$. To see this, let $p,p'\in X$ with $p\preceq p'$ and $p\in V\cap X$. Clearly, $p\in V$ and $V$ is an upward closed set, i.e., $p'\in V$. Moreover, $p'\in V \cap X$.
  Secondly, we note that the function $\tilde\mu$ is strict since the function $\mu$ is strict and $\emptyset=\emptyset\cap X$.

  Next, we show that $\tilde\mu$ is order preserving. Let $V\subseteq V'$, for some $V,V'\in\mcal S_{X^\infty}$. Then $V \cap X \subseteq V' \cap X$. And the order preservation of $\mu$ gives $\tilde\mu(V) =\mu(V\cap X) \leq \mu(V'\cap X) = \tilde\mu(V')$.

  The modularity of $\tilde\mu$ follow directly from the modularity of $\mu$ since the operator $\cap$ distributes over union, i.e., for $\Box \in \{\cup,\cap\}$ and $V,V'\in\mcal S_{X^\infty}$ we have
  \[(V \cap X) \ \Box \ (V' \cap X) = (V \Box V') \cap X.\]

  Lastly, to show that $\tilde\mu$ is Scott-continuous, let $(V_i)_{i\in I}$ be a directed family of open sets with each $V_i\in\mcal S_{X^\infty}$. Then, we find that the family $(V_i\cap X)_{i\in I}$ is a directed family of open sets in $X$. Moreover,
  \[
  \tilde\mu(\bigcup_{i\in I} V_i)
    =\mu((\bigcup_{i\in I} V_i)\cap X)
    =\mu(\bigcup_{i\in I} (V_i \cap X))
    =\sup_{i\in I} \mu(V_i \cap X)
    =\sup_{i\in I} \tilde\mu(V_i).\]
  For local finiteness, we assume that $X$ is a simple prefix order. Then, we claim that the set $\scott U=U \cup \{\infty_C \mid \exists_{p\in U}\ p \preceq \infty_C\}$ is Scott-open, whenever $U\in\mcal O_X$. From this local finiteness of $\tilde\mu$ directly follows from the local finiteness of $\mu$.

  So let $U$ be Alexandroff-open set, i.e., $U\in\mcal O_X$. Clearly, $\scott U$ is upward closed. Now it remains to show that $U$ is inaccessible by directed joins. Let $D\subseteq X^\infty$ be a directed subsets of paths with $\scott p = \sup D$ and $\scott p \in U$. Then, we distinguish the following two cases:
  \begin{enumerate}
    \item Let $\scott p\in X$. Notice that every directed set $D\subseteq X$ is a chain (Proposition~\ref{prop:irr-downclosed}). Moreover, only finite chains in the ordered set $X$ has suprema. And since in a finite chain the supremum coincides with maximum, we find that $\scott p\in D$. Thus, $D \cap \scott U \neq \emptyset$.
    \item Let $\scott p=\infty_C$, for some non-sober set $C$ in $X$. Then, by construction we find some $p'\in U$ such that $p'\preceq \infty_C$. Moreover, we have $p \preceq \infty_C$, for all $p\in D$. Thus, by downward totality of paths (cf. Lemma~\ref{lemma:dtotal-lpaths}) we find
        \[\forall_{p\in D}\ p \preceq p' \lor p' \preceq p.\]
        We claim that $\exists_{p\in D} \neg (p \preceq p')$. Suppose otherwise that $\forall_{p\in D}\ p \preceq p'$. Then, we find $\infty_C= \sup D \preceq p'$. I.e., $p'=\infty_C$. But $p'\in X$. Thus, a contradiction. Hence, $p' \preceq p$, for some $p\in D$. And $U$ is upward closed, i.e., $p\in U$. Therefore, $D \cap \scott U \neq \emptyset$.
  \end{enumerate}
  Thus, for every open set $U\in\mcal O_{X}$, the set $\scott U$ is Scott-open.
\end{proof}

\begin{theorem}[Theorem~\ref{thm:bs-contained}]
  For countable sets $A$ and $X$, the sigma-algebra $\Sigma_{\paths X}$ is contained in the Borel sigma-algebra induced by Scott-open sets, i.e., $\Sigma_{\paths X} \subseteq \Sigma_{\lpaths X}=\mcal B(\mcal S_{\lpaths X})$. Moreover, we also have $\Sigma_{\spaths X} \subseteq \Sigma_{\slpaths{X}}$.
\end{theorem}
\begin{proof}
  It suffices to show that $\paths X$ is measurable in the space $\lpaths X$, i.e., $\paths X \in \Sigma_{\lpaths X}$. From this result follows directly by constructing the trace sigma-algebra restricted to $\paths X$, i.e., $\Sigma_{\lpaths X}|_{\paths X}=\{U\cap \paths X \mid U\in\Sigma_{\lpaths X}\}$. Since $\paths X$ is measurable we have $\Sigma_{\lpaths X}|_{\paths X} \subseteq \Sigma_{\lpaths X}$. Moreover, every open set $U\in \mcal O_{\paths X}$ is contained in $\Sigma_{\lpaths X}|_{\paths X}$ because $U=\scott U \cap \paths X$ and $\scott U\in \Sigma_{\lpaths X}$; thus, $\Sigma_{\paths X}\subseteq \Sigma_{\lpaths X}|_{\paths X}$.

  Thus, it suffices to show that the set $\paths X$ is measurable in $\lpaths X$.
  Next, we claim that every singleton $\{p\}$ (for $p\in\paths X$) is measurable in the space $\lpaths X$, i.e., $\{p\}\in\Sigma_{\lpaths X}$, for every $p\in \paths X$. To see this, first note that $\{p\}=\history p \cap \scott{(\future p)}$. From Theorem~\ref{thm:val-scott-extension} we know that $\scott{(\future p)}$ is a Scott-open. Clearly, the history $\history p$ is a Scott-closed set. Moreover, every open and closed set are Borel sets and since sigma-algebra is closed under countable intersection, we find that the singleton $\{p\}\in\Sigma_{\lpaths X}$ is Borel measurable.

  Now to complete the proof, it is sufficient to show that the Alexandroff open set $\future{\eseq x}=\{p\in \paths X \mid \eseq x \preceq p\}$ is countable (for $x\in X$). From this it follows that the set $\paths X$ is Borel measurable $\paths X\in \Sigma_{\lpaths X}$ because the set $\paths X = \bigcup_{x\in X} \future{\eseq x}$ is countable (since countable unions of countable sets is countable). 

  Thus, it remains to verify that the set $\future {\eseq x}$ is countable. Consider the following sets
  \[\future{(\eseq x,n)}=\{p\in\future{\eseq x}\mid |\max \dom p|=n\},\]
  where $|\sigma|$ denote the length of the word $\sigma\in\tact^*$. Clearly, $\future{\eseq x}=\bigcup_{n\in\mathbb N}\future{(\eseq x,n)}$. Now, by induction on $n$ we prove that each set $\future{(\eseq x,n)}$ is countable. Suppose that the set $\future{(\eseq x,n)}$ is countable. Then, a path in $\future{(\eseq x,n)}$ can be extended to only countably many paths in $\future{(\eseq x,n+1)}$ because the sets $\act$ and $X$ are. Let $p\in\paths X$. For $a\in\tact,x\in X$, we write $p_{a,x}$ to denote the path that extends $p$ such that $p_{a,x}(\max \dom p a)=x$. Then, the set $\{p_{a,x} \mid a\in \tact \land x\in X\}$ (for a given $p$) is isomorphic to the countable set $A\times X$ (Cartesian product of countable sets is countable). Therefore, the set $\future{(\eseq x,n+1)}=\bigcup_{p\in\future{(\eseq x,n)}} \{p_{a,x} \mid a\in\tact \land x\in X\}$ is countable. Thus, the set $\future {\eseq x}$ is also countable (since countable union of countable sets is a countable set).

  In the context of stutter equivalent paths, we show that $\spaths X\in \Sigma_{\slpaths{X}}$. First observe that every singleton $\{\qset{[p]}\}$ is measurable in $\slpaths{X}$ because $\{\qset{[p]}\}=\history{\qset{[p]}} \cap
  \scott{(\future{\qset{[p]}})}$. Lastly, $\paths X=\bigcup_{p\in\paths X} \{\qset{[p]}\}$ is measurable since countable union of measurable sets is measurable.
\end{proof}

\section{Proofs concerning Section~\ref{sec:prob}}\label{app-prob}
%

\begin{proposition}[Proposition~\ref{prop:separated}]
  In an Alexandroff space $(X,\mcal O_X)$, a separated subset $U\subseteq X$ (i.e., $U=U^\star$) has topologically distinguishable points. Moreover, 
  \begin{enumerate}
    \item the separation closure of a set is always separated, i.e., $U^\star = {U^\star}^\star$, for any $U\subseteq X$.
    \item the collection of separated sets is hereditary, i.e., if $U_1\subseteq U_2$ and $U_2=U_2^\star$, then $U_1=U_1^\star$.
    \item for any subset $U\subseteq X$, we have $(\future U)^\star\subseteq U^\star$. Moreover the converse also holds, if the underlying space $X$ is a $T_0$ space. Here, upward closure is w.r.t. the specialisation order $\preceq$, i.e., $x \preceq x' \iff \closure x \subseteq \closure{x'}$, for any $x,x'\in X$.
  \end{enumerate}
\end{proposition}
\begin{proof}
  Let $U=U^\star$ be a separated set and let $x,x'\in U$ with $x\neq x'$. Then, $\closure x \cap U = \{x\}$ and $\closure{ x'} \cap U =\{x'\}$. Thus, $x'\not\in \closure x,x\not\in\closure {x'}$. Clearly, the points $x,x'$ are topologically distinguishable because there is an open set $X\setminus \closure x$ which contains $x'$, but not $x$, i.e., $x'\in X\setminus \closure x$ and $x\not\in X\setminus \closure x$.

  \begin{description}
    \item[1] Let $U\subseteq X$ and let $x\in U^\star$. Then, we show that $\closure x\cap U^\star =\{x\}$. Suppose otherwise that there is a point $x'\in\closure x\cap U^\star$ with $x\neq x'$. Then, $x'\in U^\star$. Moreover, from the above we find $x'\not\in \closure x$; but this violates the assumption $x'\in \closure x$. Therefore, $x\in {U^\star}^\star$. The other direction ${U^\star}^\star \subseteq U^\star$ is trivial.
    \item[2] Let $U_1 \subseteq U_2$ and $U_2 = U_2^\star$. Then, it suffices to show that $U_1\subseteq U_1^\star$, i.e., $\closure x \cap U_1 =\{x\}$, for any $x\in U_1$. Let $x\in U_1$. Then, we find $\closure x \cap U_1 \subseteq \closure x \cap U_2 =\{x\}$. Clearly, $x\in \closure x \cap U_1$. Thus, $\closure x \cap U_1 =\{x\}$; whence, $x\in U_1^\star$.
    \item[3] Let $x\in(\future U)^\star$. Then, $\closure x \cap \future U = \{x\}$. I.e., $y \preceq x$ for some $y\in U$. And by the definition of specialisation order we find $\closure y \subseteq \closure x$ (i.e., $y\in\closure x$). But $\closure x \cap \future U = \{x\}$, thus, $y=x$ and $x\in U$. Moreover, $\closure x \cap U \subseteq \closure x \cap \future U=\{x\}$; whence, $x\in U^\star$.

        For the other direction, let $x\in U^\star$. Clearly, $x\in \future U$ because $x\in U$. Thus, it remains to show that $\closure x \cap \future U = \{x\}$. Suppose $y\neq x$ and $y\in \closure x \cap \future U$, i.e., $y\in \closure x$ and $y\in \future U$. Then, we find some $x'\in U$ such that $y\in \future {x'}$ and $x'\in U$. By transitivity of the specialisation order we find, $x' \preceq y \preceq x$. Moreover, $x\in U^\star$; thus, $x'=x$. And, since the given space $X$ is $T_0$, we know that $\preceq$ is antisymmetric. Thus, $y=x$, which contradicts the assumption $y\neq x$; whence $x\in (\future U)^\star$.\qedhere
  \end{description}
\end{proof}

\begin{lemma}[Lemma~\ref{lemma:step-sepclosure}]
  For a given probabilistic transition system $(X,\tact,P)$, define a set as follows:
  $x \steps L Y =\{p\in\paths X \mid p(\eseq{})=x\land \trace p\in L\land \last p\in Y \}.$
  Then, $\sum_{p\in x \step L Y}\mu_P(p)= \sum_{p\in (x \steps L Y)^\star} \mu_P(p) $.
\end{lemma}
\begin{proof}
We begin by proving that $x \step L Y \subseteq (x \steps L Y)^\star$. Let $p\in x \step L Y$. Then, we need to show that $\mcal N(p) \cap x \steps L Y = \{p\}$. Suppose otherwise that $\mcal N(p) \cap x \steps L Y \neq \{p\}$. Note that $p\in \mcal N(p)=\history p$ and $p\in x\steps L Y $ (because $p\in x\step L Y$). Therefore, $\{p\}\subseteq \mcal N(p)\cap x \steps L Y$. Now consider a path $q\in \history p \cap q\in x\steps L Y$ with $p\neq q$. Then, we find $q \prec p$, $\trace q \in L$, and $\last q\in Y$. Moreover, $q\prec p$ and $p(\eseq{})=x$ implies that $q(\eseq{})=x$, i.e., $q\in\exec x$. Thus, $p\not\in x \step L Y$, which is a contradiction. Hence, $\sum_{p\in x \step L Y} \mu_P(p)\leq \sum_{p\in (x \steps L Y)^\star} \mu_P(p) $.

Now consider a path $p\in (x\steps L Y)^\star$. We will first show that if $p\in\exec x$ then $p\in x \step L Y$. Let $p\in (x \steps L Y)^\star$ with $p\in\exec x$. Then, we need to show that there is no path $q$ such that $q\prec p$, $\trace q \in L$, and $\last q\in Y$. Suppose otherwise that there is a $q$ with $q\prec p$, $\trace q \in L$, and $\last q\in Y$. Then, $q\in x \steps L Y$. But, $\mcal N(p) \cap x \steps L Y \neq \{p\}$, since $q\in \mcal N(p)$ and $q\in x \steps L Y$. Thus, $p \not\in (x\steps L Y)^\star$, which is a contradiction. Hence, $(x\steps L Y)^\star \cap \exec x = x \step L Y$.
Moreover, for any $p\in (x\steps L Y)^\star$ and $p\not\in\exec x$ we have $\mu_P(p)=0$. Thus, $\sum_{p\in x \step L Y} \mu_P(p) =\sum_{p\in (x \steps L Y)^\star} \mu_P(p)$.
\end{proof}

\begin{theorem}[Theorem~\ref{thm:sep:prob-char}]
Let $p\in\paths{X}$ be a path and let $U\subseteq\paths X$ be a separated set of paths such that $p\preceq U$, i.e., $\forall_{q\in U}\ p\preceq q$. Then, $\sum_{q\in U}\mu_P(q)\leq\mu_P(p)$.
\end{theorem}
\begin{proof}
We first compute:
{\allowdisplaybreaks
\begin{align*}
&\sum_{q\in U}\mu_P(q)=\sum_{q\in U\cap\bigcup\{\exec x\mid x\in X\}}\prod_{\sigma a\in\dom q}P(q(\sigma),y,q(\sigma a))\\
=&\sum_{q\in U\cap\bigcup\{\exec x\mid x\in X\}}\prod_{\sigma a\in\dom p}P(p(\sigma),a,P(\sigma a))\cdot\prod_{\sigma a\in\dom q \setminus\dom p}P(q(\sigma),a,q(\sigma a))\\
=&\sum_{q\in U\cap\bigcup\{\exec x\mid x\in X\}}\mu_P(p)\cdot\prod_{\sigma a\in\dom q \setminus\dom p}P(q(\sigma),a,q(\sigma a))\\
=&\mu_P(p)\cdot\sum_{q\in U\cap\bigcup\{\exec x\mid x\in X\}}\prod_{\sigma a\in\dom q\setminus\dom p}P(q(\sigma),a,q(\sigma a)).
\end{align*}
}
Now asume $U$ is finite -- if it is not, pick arbitrarily a finite $\overline U\fsubseteq U$, the result for an infinite $U$ follows then by definition of infinite sums -- then, by separation, there exists an $n\in\mathbb N_0$, such that for all $q_1, q_2\in U\cap\bigcup\{\exec x\mid x\in X\}$ it holds that $q_1\neq q_2\Rightarrow\dom{q_1}\cap\dom{q_2}\cap \tact^n=\emptyset$. From $\sum_{(a,x')\in\tact\times X}P(x,a,x')\in\{0,1\}$ and non-negativity of $P$ it follows for any path $q\in U\cap\bigcup\{\exec x\mid x\in X\}$ that $$\prod_{\sigma a\in\dom{q}\setminus\dom p}P(q(\sigma),a,q(\sigma a))\leq\prod_{\sigma a\in(\dom q\setminus\dom p)\cap\tact^{\leq n}}P(q(\sigma),a,q(\sigma a))$$
We will now show that, in turn $$\sum_{q\in U\cap\bigcup\{\exec x\mid x\in X\}}\prod_{\sigma a\in(\dom q\setminus\dom p)\cap\tact^{\leq n}}P(q(\sigma),a,q(\sigma a))\leq1.$$ Plugging this into the previous computation yields the desired result.

In order to see that $$\sum_{q\in U\cap\bigcup\{\exec x\mid x\in X\}}\prod_{\sigma a\in(\dom q\setminus\dom p)\cap\tact^{\leq n}}P(q(\sigma),a,q(\sigma a))\leq1$$ holds, we model the weights each path $q\in U\cap\bigcup\{\exec x\mid x\in X\}$  contributes to the sum using a rooted, directed tree $G=(V,E)$, where $V$ is the set of vertices and $E$ is the set of edges. In this tree, each vertex is labelled with a path $q$ and  and will be associated with a weight $w(q)$ that is the sum of all weights of $q'\in U\cap\bigcup\{\exec x\mid x\in X\}$ in the above sum, i.e. it holds that $$w(q)=\sum_{q'\in U\cap\bigcup\{\exec x\mid x\in X\}, \dom{q'}\supseteq\dom q}\prod_{\sigma a\in(\dom{q'}\setminus\dom q)\cap\tact^{\leq n}}P(q'(\sigma),a,q'(\sigma a)).$$
For this purpose, build the tree from the root downwards as follows. The root is labelled $p$. Now, given a vertex with label $q$, if $q\in U$, $q$ is a leaf, otherwise consider all pairs $(a,x)\in \tact\times X$ such that there exists a $q'\in U$ with $\max\dom{q}a\in\dom{q'}$ and $q'\mid_{\dom{q}}=q$. For all such $(a,x)$, there exists a successor vertex labelled $q'\mid_{\history{\max\dom{q}a}}$ and the edge $e=(q, q'\mid_{\history{\max\dom{q}a}})$ has weight $$w(e)=P(q'(\max\dom{q}),a,q'(\max\dom{q}a)).$$ Note, that if $U$ is non-empty, this is a mutually exclusive choice, because of separetedness of $U$; if $U$ is empty, then the sum is taken over an empty set, and the result holds trivially. Now, given any vertex $v$, we can inductively define its weight $w(v)$ via a case distinction. If $v$ is a leaf, then $w(v)=1$. Otherwise, $$w(v)=\sum_{\{v'\mid (v,v')\in E\}}w(v,v')\cdot w(v).$$ This way, we ensure that each node labelled as $q$ has the weight $$w(q)=\sum_{q'\in U\cap\bigcup\{\exec x\mid x\in X\}, \dom{q'}\supseteq\dom q}\prod_{\sigma a\in(\dom{q'}\setminus\dom q)\cap\tact^{\leq n}}P(q'(\sigma),a,q'(\sigma a)).$$ In particular, it holds that $$w(p)=\sum_{q\in U\cap\bigcup\{\exec x\mid x\in X\}}\prod_{\sigma a\in(\dom{q}\setminus\dom p)\cap\tact^{\leq n}}P(q(\sigma),a,q(\sigma a)).$$ Now it is easy to see inductively, that for each vertex $q$ it holds that $w(q)\leq1$. If $q$ is a leaf, this follows by definition. Otherwise, we can see inductively, that all successor vertices have a weight less than or equal $1$. Now, let $x=q(\max\dom q)$, then the weight of $q$ must be smaller than or equal to $\sum_{(a,x')\in\tact\times X}P(x,a,x')$ -- each edge starting in $q$ corresponds to one such $P(x,a,x')$ and all these values get multiplied with a value that is not greater than $1$ -- and by definition of $P$, this smaller or equal to $1$, which concludes the proof.
\end{proof}

\begin{theorem}[Theorem~\ref{thm:pts-valuation}]
  Let $(X,\tact,P)$ be a given probabilistic transition system. Then, the function $\mcal O_{\paths X} \rTo^{\tilde\mu_P} \nnreals$ defined as $\tilde\mu_P(U)=\mu_P(U^\star)$ (for every open set $U\in\mcal O_{\paths X}$) is a locally finite and Scott-continuous valuation.
\end{theorem}
  \begin{proof}
  Clearly, $\tilde\mu_P$ is strict because $\emptyset^\star=\emptyset$ and $\mu_P(\emptyset)=0$. Also, local finiteness is straightforward because for any path $p$ we have a finitely valued open neighbourhood $\future p$. Moreover, using Proposition~\ref{prop:separated}.\ref{prop:separated:future}  we find that $\tilde\mu_P(\future p) = \mu_P((\future p)^\star)=\mu_P(\{p\}^\star)=\mu_P(p)<\infty$.

  \begin{description}
    \item[Order preservation] Let $U_1\subseteq U_2$ for some $U_1,U_2\in\mcal C_{\paths X}$. Let $p\in U_1^\star$. I.e., $p\in U_2$. Since $U_2^\star$ contains all the minimal elements of $U_2$, we find some $q\in U_2^\star$ such that $q\preceq p$. Thus, we can partition the set $U_1^\star$ in the following way: $p \approx p' \iff \exists_{q\in U_2^\star}\ q \preceq p \land q\preceq p'$. To verify transitivity of $\approx$, let $p\approx p'$ and $p'\approx p''$. Then we find $q\preceq p,q\preceq p',q'\preceq p'$, and $q'\preceq p''$, for some $q,q'\in U_2^\star$. By downward totality of paths we have $q \preceq q' \vee q'\preceq q$. In any case, $q=q'$ since $q,q'\in U_2^\star$. Thus, $p\approx p''$. Next, we prove that for any $q\in U_2^\star$ and $[p]_q=\{p'\in U_1^\star \mid p\approx p'\}$ with $p\in U_1^\star$ and $q\preceq p$ we have $\mu_P(q)\geq \sum_{p'\in[p]_q} \mu_P(p')$. Note that each $p'\in [p]_q$ is a future of $q$, i.e., $q\preceq p'$ and any two paths in $p',p''\in[p]_q$ are topologically distinguishable paths (i.e., in order theoretic terms they are incomparable). Furthermore, probabilistic transition system are image finite, i.e., there are finitely many elements $p'\in[p]_q$ with $\mu_P(p')\geq 0$ and the sum of outgoing transitions is either 0 or 1. I.e., there are some finitely many real numbers $r_1,\cdots,r_n\in (0,1)$ such that $\sum_{i=1}^n r_i\leq 1$ and $\sum_{p'\in [p]_q} \mu_P(p')=\sum_{i=1}^n r_i \mu_P(q)\leq \mu_P(q)$. Thus, $\mu_P(U_1^\star) \leq\mu_P(U_2^\star)$. The proof for the case when $U_1,U_2$ are infinite sets follows immediately from above.
    \item[Modularity] Let $U_1,U_2\in\mcal O_{\paths X}$ be any two open sets. Then we find
  \begin{align*}
    \tilde\mu_P(U_1) + \tilde\mu_P(U_2) =& \ \mu_P(U_1^\star) +\mu_P(U_2^\star)\\
    =& \ \sup \{\mu_P(\bar U_1) \mid \bar U_1 \fsubseteq U_1^\star\} + \sup \{\mu_P(\bar U_2) \mid \bar U_2 \fsubseteq U_2^\star\}\\
    =& \ \sup \{\mu_P(\bar U_1) + \mu_P(\bar U_2)\mid \bar U_1 \fsubseteq U_1^\star\land \bar U_2 \fsubseteq U_2^\star\}.
  \end{align*}
  and
  \begin{align*}
    \tilde\mu_P(U_1 \cup U_2) +& \tilde\mu_P(U_1 \cap U_2) = \mu_P((U_1\cup U_2)^\star) + \mu_P((U_1\cap U_2)^\star)\\
    =& \ \sup \{\mu_P(\bar U) \mid \bar U\fsubseteq (U_1\cup U_2)^\star\} + \sup \{\mu_P(\bar U') \mid \bar U' \fsubseteq (U_1\cap U_2)^\star\}\\
    =& \ \sup \{\mu_P(\bar U) + \mu_P(\bar U') \mid \bar U \fsubseteq (U_1\cup U_2)^\star \land \bar U' \fsubseteq (U_1\cap U_2)^\star\}.
  \end{align*}
  Let $\bar U_1 \fsubseteq U_1^\star$ and $\bar U_2 \fsubseteq U_2^\star$. Let $\bar U\subseteq (U_1 \cup U_2)^\star$ and $\bar U' \subseteq (U_1 \cap U_2)^\star$ be the smallest sets satisfying the condition (respectively):
  \[\forall_{p\in \bar U_1 \cup \bar U_2}\exists_{q\in\bar U}\ q\preceq p \qquad \text{and}\qquad \forall_{p\in \bar U_1 \cap \bar U_2} \exists_{q\in\bar U'}\ q\preceq p.\]
  Then, we find that the sets $\bar U,\bar U'$ are finite since the given sets $\bar U_1\cup\bar U_2,\bar U_1\cap \bar U_2$ are finite. Next, we prove the inequality $\mu_P(\bar U_1) + \mu_P(\bar U_2) \leq \mu_P(\bar U) + \mu_P(\bar U')$ indirectly by establishing:
  \[\mu_P(\bar U_1) + \mu_P(\bar U_2) \leq  \mu_P(\bar U \cup \bar U').\]
  Let $p_1\in\bar U_1 ,p_2 \in\bar U_2$. I.e., there are paths $q_1,q_2\in \bar U$ such that $q_1 \preceq p_1$ and $q_2 \preceq p_2$. If $q_1 \neq q_2$ then we are done. So suppose $q_1 = q_2$. Then, we find that $p_1,p_2\in U_1 \cap U_2$ because $U_1,U_2$ are open sets. I.e., there are paths $q_1',q_2'\in \bar U'$ such that $q_1' \preceq p_1$ and $q_2' \preceq p_2$. Notice that $q_1'\preceq p_1$ and $q_1 \preceq p_1$. So by downward totality we get $q_1 \preceq q_1' \lor q_1' \preceq q_1$. Consider the case $q_1'\preceq q_1$. Then, we find $q_1'\in U_1 \cup U_2$ and $q_1=q_1'$ since $q_1\in (U_1\cup U_2)^\star$. Therefore, the above condition simplifies to $q_1 \preceq q_1'$. Likewise, we can derive $q_1 \preceq q_2'$. Thus, $\mu_P(p_1) + \mu_P(p_2) \leq \mu_P(q) + \mu_P(q_1') + \mu_P(q_2')$.

  For the other direction, let $\bar U\fsubseteq (U_1 \cup U_2)^\star$ and $\bar U' \fsubseteq (U_1 \cap U_2)^\star$. Let $\bar U_1\subseteq U_1^\star$ and $\bar U_2\subseteq U_2^\star$ (respectively) be the smallest sets satisfying the conditions:
  \[\forall_{p\in(\bar U\cap U_1)\cup \bar U'}\exists_{q\in \bar U_1}\ q \preceq p \qquad \text{and} \qquad \forall_{p\in(\bar U\cap U_2)\cup \bar U'}\exists_{q\in \bar U_2}\ q \preceq p.\]
  Then, the sets $\bar U_1,\bar U_2$ are finite just because the given sets $\bar U,\bar U'$ are finite. Next, we prove the inequality: $\mu_P(\bar U) + \mu_P(\bar U') \leq \mu_P(\bar U_1) + \mu_P(\bar U_2)$. Let $p\in\bar U,q\in\bar U'$. Then, we find $p\in U_1\cup U_2$. So let $p\in U_1$ (the case when $p\in U_2$ is symmetric). Then, we find $p\in \bar U\cap U_1$. Clearly, $p,q\in (\bar U\cap U_1) \cup \bar U'$. I.e., there are some paths $p',q'\in \bar U_1$ such that $p'\preceq p$ and $q' \preceq q$. If $p'\neq q'$, then we are done. So suppose $p'=q'$. Then, we observe that $q\in (\bar U\cap U_2) \cup \bar U'$ (because $q\in \bar U'$). I.e., there is some $q''\in \bar U_2$ such that $q''\preceq q$. Thus, $\mu_P(p)+\mu_P(q) \leq \mu_P(p')+\mu_P(q'')$.
    \item[Scott-continuity] Let $(U_i)_{i\in I}$ be a directed family of open sets. Then, we have to show that $\tilde\mu_P(\bigcup_{i\in I} U_i) = \sup_{i\in I} \tilde\mu_P(U_i)$. First, we note that $U_i \subseteq \bigcup_{j\in I} U_j$ (for each $i\in I$). Thus, by order preservation of $\tilde\mu_P$ we get $\tilde\mu_P(U_i) \leq \tilde\mu_P(\bigcup_{j\in I} U_j)$, for each $i\in I$. And by the universal property of $\sup$ we find that $\sup_{i\in I}\tilde\mu_P(U_i) \leq \tilde\mu_P(\bigcup_{i\in I} U_i)$.

        Therefore, it remains to show that $\tilde\mu_P(\bigcup_{i \in I} U_i) \leq \sup_{i\in I}\tilde\mu_P(U_i)$. I.e., we have to prove
        \[
        \sup\left\{\mu_P(\bar U) \mid \bar U \fsubseteq \Big(\bigcup_{i\in I} U_i\Big)^\star\right\}
        \leq
        \sup\left\{\mu_P(U_i^\star) \mid i\in J \land J\fsubseteq I\right\}.
        \]
        Let $\bar U\fsubseteq(\bigcup_{i\in I}U_i)^\star$ and $p\in\bar U$. Then there exists an $i\in I$ such that $p\in U_i=:U_p$. Since $\bar U$ is finite, we can write $\bar U=\{p_1, p_2, ..., p_n\}$ for some natural number $n\geq 1$. We now inductively define a \emph{finite family} of open sets $V_m$ as follows: $V_{p_1}=U_{p_1}$ and given an $m\geq 1$, there exists a set $V_{p_{m+1}}$ such that $V_{p_{m+1}}\supseteq V_{p_m}\cup U_{p_{m+1}}$ (since the given family $(U_i)_{i\in I}$ is directed). It is clear that $\bar U$ is separated (cf. Proposition~\ref{prop:separated}.\ref{prop:separated:hereditary}). Thus, we can conclude
        \[\bar U\subseteq V_{p_n}
        \implies
            \tilde\mu_P(\bar U)\leq\tilde\mu_P(V_{p_n})
                \implies
                    \mu_P(\bar U)\leq\mu(V_{p_n}^*)\enspace.\qedhere\]
  \end{description}
\end{proof}
\begin{proposition}[Proposition~\ref{prop:alpha-measure}]
  The mapping defined in \eqref{eq:measure-alpha} induced by a probabilistic transition system is a probability measure.
\end{proposition}
\begin{proof}
  Clearly, $\alpha(x)(\emptyset)=\tilde\mu_P(\emptyset \cap \future{ \eseq x}) = \tilde\mu_P(\emptyset)=0$. Now, let $(U_i)_{i\in I}$ be a countable family of pairwise disjoint Borel subsets $U_i\in \Sigma_{\spaths{X}}$ (for each $i\in I$). Then, we claim that for any $i,j\in I$ with $i\neq j$ we have $\inv{\pi_X}(U_i)\cap \inv{\pi_X}(U_j)=\emptyset$. Suppose otherwise we have $p\in \inv{\pi_X}(U_i)\cap \inv{\pi_X}(U_j)$. Then, $\qset{[p]}\in U_i\cap U_j$ violating the aforementioned assumption. Thus,
  {\allowdisplaybreaks
  \begin{align*}
    \alpha(x)(\bigcup_{i\in I} U_i)
        &=\tilde\mu_P\Big(\inv{\pi_X}\big(\bigcup_{i\in I} U_i\big) \cap \future{\eseq x}\Big)\\
        &=\tilde\mu_P\Big(\big(\bigcup_{i\in I} \inv{\pi_X}(U_i)\big) \cap \future{\eseq x}\Big)\\
        &=\tilde\mu_P\Big(\bigcup_{i\in I} \big(\inv{\pi_X}(U_i) \cap \future{\eseq x}\big)\Big)\\
        &=\sum_{i\in I}\tilde\mu_P\big
        (\inv{\pi_X}(U_i)\cap\future{\eseq x}\big)=\sum_{i\in I}\alpha(x)(U_i).
  \end{align*}
  }
  Lastly, $\alpha(x)(\spaths X)=\tilde\mu_P(\paths X\cap \future{\eseq x})=\tilde\mu_P(\future{\eseq x})=\mu_P(\eseq x)=1$.
\end{proof}

\begin{theorem}[Theorem~\ref{thm:pbisim}]\label{thm:apbisim}
  Two states are probabilistic delay bisimilar if and only if they are $\mcal G\circ \tspaths$-behaviourally equivalent.
\end{theorem}
\begin{proof}
\fbox{$\Rightarrow$} For this direction, it was left to show that $\nu_{f(x)}$ in \eqref{eq:qconst} is well-defined. We first recall a technical result from \cite[Lemma~24]{sokolova:2009:sacs}. For any branching bisimulation $R$, any $R$-saturated set $M\subseteq X$, and any saturated block $W\subseteq \tact^*$, the following equation holds
  \begin{equation}\label{eq:completeness}
    P(x,W,M)=P(x',W,M), \quad \text{for any $x,x'\in X$ such that $x Rx'$}.
  \end{equation}
  A subset $M\subseteq X$ is $R$-saturated if for all $x\in M$, the whole equivalence class of $x$ is contained in $M$, i.e., $\forall_{x,x'}\ (x\in M \land x R x') \implies x'\in M$. A subset $W \subseteq \tact^*$ is saturated if it can be decomposed into a collection of blocks $B_i$, for some index set $I$, i.e., $W=\bigcup_{i\in I} B_i$, where each block $B_i=\tau^*a_1\tau^*\cdots\tau^*a_n\tau^*$, for a word $a_1a_2\cdots a_n\in \tact^*$.

  Let $U=\inv{\pi}(\inv{\spaths f} V)$. First we claim that $\trace U=\{\trace p\mid \qset{[fp]}\in V\}$ is saturated by establishing the equation:
  \[\trace U = \bigcup_{\qset{[fp]}\in V} g(\trace{\stutter p}),\]
  where the operation $\tact^* \rTo^g \mcal P(\tact^*)$ is defined inductively as follows: $g(\eseq{}) =\tau^*$ and $g(\sigma a) = g(\sigma) \tau^* a\tau^*$. The direction $\trace U \subseteq \bigcup_{\qset{[fp]}\in U} g(\trace{\stutter p})$ because $\sigma \in g(\sigma)$, for any $\sigma\in\tact^*$. For the other direction, let $\sigma\in g(\trace{\stutter{p}})$, for some $\qset{[fp]}\in V$. Let $\trace{\stutter{p}}=a_1a_2\cdots a_n$ (for some $n\geq 1$). Then, we find $\sigma=\tau^{m_1}a_1\tau^{m_2}\cdots \tau^{m_{n}} a_n \tau^{m_{n+1}}$.

  Define a path $q$ whose domain is $\sigma$ as follows:
  \[
  q(\sigma')=
  \begin{cases}
    \stutter{p}(\eseq {}) & \text{if}\ \sigma'\preceq \tau^{m_i},\\
    \stutter{p}(a_1\cdots a_i) & \text{if}\ \tau^{m_1}a_1\cdots \tau^{m_{i}}a_i \prec \sigma' \preceq \tau^{m_1}a_1\cdots \tau^{m_{i}}a_i\tau^{m_i+1}\land i\leq n,\\
    \stutter{p}(a_1\cdots a_i) & \text{if}\ \sigma'=\tau^{m_1}a_1\cdots \tau^{m_{i}}a_i \land i\leq n.
  \end{cases}
  \]
  Clearly, $\stutter{p}=\stutter{q}$, i.e., $p\sim q$. Thus, $\qset{[fp]}=\qset{[fq]}\in V$; hence, $q\in \trace U$.

  Next, we verify that $\last U$ is an $R$-saturated set. To see this let, $x\in \last U$ and $x R x'$. Then, there is a path $p\in U$ such that $\last p=x$ and $\qset{[fp]}\in V$. Suppose $\dom p=\eseq{}$. Then, we find that $f\circ \eseq x=f\circ \eseq{x'}$ (because $x Rx'$). I.e., $\eseq{x'}\in U$. Now suppose $p$ has a non-empty domain, i.e., there is some path $p'\prec p$ such that $\max \dom {p'}=\sigma$, $\max \dom p=\sigma a$, for some $\sigma\in\tact^*,a\in\tact$. Then, we define a path $q$ whose domain is identical to $p$ as follows: $q(\sigma')=p'(\sigma')$ if $\sigma'\in\dom {p'}$ and $q(\sigma a)=x'$. Again, we find that $f\circ p =f\circ q$ because $x R x'$. Thus, $q\in U$, i.e., $x'\in U$.

  Now we can prove that $\nu_{f(x)}$ in \eqref{eq:qconst} is well-defined. Suppose $x R x'$. Then it suffices to show that $\tilde\mu(U\cap \eseq x)=\tilde\mu(U\cap \eseq x')$, where $U=\inv{\pi}(\inv{\spaths f} V)$ and $V\in\mcal O_{\spaths{X/R}}$. This follows directly from the fact that $\tilde\mu(U\cap \eseq x)=\mu((U\cap \eseq x)^\star)=P(x,\trace U,\last U)$ and Equation~\ref{eq:completeness}.

  (\fbox{$\Leftarrow$})
  Let $(X,\tact,P)$ be a fully probabilistic system and $(X,\alpha)$ be the corresponding $\mcal G\circ \tspaths$-coalgebra. Moreover, let $(Y,\beta)$ be a $\mcal G\circ \tspaths$-coalgebra and $X \rTo^f Y$ be a $\mcal G\circ \tspaths$ coalgebra homomorphism.
  We will show that the equivalence relation $x R x' \iff f(x)=f(x')$ is a probabilistic branching bisimulation in two stages. First, we claim that
  \begin{equation}
        \label{eq:1}
          \beta(f(x))\Big(\bigcup_{q\in f(x) \steps{\hat a} \{y\}} \future {\qset {[q]} }\Big)=
            \alpha(x)\Big(\bigcup_{p\in x \steps{\hat a} \inv f\{y\}} \future{\qset{[p]}}\Big) \qquad \text{for any $x\in X,y\in Y$ }.
        \end{equation}
  Here, we abbreviate $f(x) \steps{\tau^*\hat a} \{y\}$ by $f(x) \steps{\hat a} \{y\}$. Notice that the sets $\bigcup_{q\in f(x) \steps{\hat a} \{y\}} \future {\qset {[q]} }$ and $\bigcup_{p\in x \steps{\hat a} \inv f\{y\}} \future{\qset{[p]}}$ are open sets (hence Borel sets).
  Second, we claim that
  \begin{equation}\label{eq:2}
    P(x,\tau^*\hat a,\inv f \{y\})=\alpha(x)\Big(\bigcup_{p\in x \steps{\hat a} \inv f\{y\}} \future{\qset{[p]}}\Big) \qquad
    \text{for any $x\in X,y\in Y$}.
  \end{equation}
  Using the above two properties we can now prove the transfer property of bisimulation as follows:
  \begin{align*}
    P(x,\tau^*\hat a,\inv f \{y\})
    =&\ \alpha(x)\Big(\bigcup_{p\in x \steps{\hat a} \inv f\{y\}} \future{\qset{[p]}}\Big) & (\text{Using Equation~\eqref{eq:2}})\\
    =&\ \beta(f(x))\Big(\bigcup_{q\in f(x) \steps{\hat a} \{y\}} \future {\qset {[q]} }\Big)& (\text{Using Equation~\eqref{eq:1}})\\
    =&\  \beta(f(x'))\Big(\bigcup_{q\in f(x') \steps{\hat a} \{y\}} \future {\qset {[q]} }\Big) & (\because f(x)=f(x'))\\
    =&\ \alpha(x')\Big(\bigcup_{p\in x' \steps{\hat a} \inv f\{y\}} \future{\qset{[p]}}\Big) & (\text{Using Equation~\eqref{eq:1}})\\
    =&\ P(x',\tau^*\hat a,\inv f \{y\}) & (\text{Using Equation~\eqref{eq:2}}).
  \end{align*}

  Next, we show the satisfaction of Property \eqref{eq:1}. For this we derive
  {\allowdisplaybreaks
  \begin{align*}
    \beta(f(x))\Big(\bigcup_{q\in f(x) \steps{\hat a} \{y\}} \future {\qset {[q]} }\Big) &=
        \tilde\mu\bigg( \inv {\paths f} \Big(\inv {\pi_Y} (\bigcup_{q\in f(x) \steps{\hat a} \{y\}} \future {\qset {[q]} }) \Big) \cap \future {\eseq x}\bigg)\\
        &=
        \tilde\mu\bigg( \inv {\paths f} \Big(\bigcup_{q\in f(x) \steps{\hat a} \{y\}} \inv {\pi_Y} (\future {\qset {[q]}}) \Big) \cap \future{\eseq x}\bigg) \tag{Lemma~\ref{lemma:stutt-hreflect}}\\
        &=
        \tilde\mu\bigg( \inv {\paths f} \Big(\bigcup_{q'\in\{q' \mid \exists_q\ q' \sim q \land q\in f(x) \steps{\hat a} \{y\}\}} \future {q'} \Big) \cap \future{\eseq x}\bigg)\\
        &=
        \tilde\mu\bigg( \Big(\bigcup_{q'\in\{q' \mid \exists_q\ q' \sim q \land q\in f(x) \steps{\hat a} \{y\}\}} \inv {\paths f}(\future {q'})\Big) \cap \future{\eseq{x}} \bigg)\\
        &=
        \tilde\mu\bigg( \Big(\bigcup_{p\in f(x) \stepsX{\hat a}\{y\}} \future {p} \Big) \cap \future{\eseq x}\bigg)\tag{where, \(f(x) \stepsX{\hat a}\{y\}=\{p \mid \exists_q\ fp \sim q \land q\in f(x) \steps{\hat a} \{y\}\}\)}\\
        &= \tilde\mu\bigg( \bigcup_{p\in f(x) \stepsX{\hat a}\{y\}} \Big(\future {p}  \cap \future{\eseq x} \Big)\bigg)\\
        &= \tilde\mu\bigg( \bigcup_{p\in f(x) \stepsX{\hat a}\{y\}} \future {p} \bigg) \\
        &=
        \tilde\mu( \bigcup_{p\in x \steps {\hat a} \inv f \{y\}} \future p).
  \end{align*}
  }
  To show that why the last two steps in the above derivation are equivalent, we prove that $f(x) \stepsX{\hat a}\{y\} = x \steps {\hat a} \inv f \{y\}$. Let $p\in f(x) \stepsX{\hat a}\{y\}$. Then we find $\eseq x \preceq p$, $fp\sim q$, and $q\in f(x) \steps{\hat a}\{y\}$, for some $q\in\paths Y$. From Lemma~\ref{blemma}\eqref{blemma:stutt-pres-startstate} and \ref{blemma}\eqref{blemma:stutt-pres-laststate} we find that $p\in x \steps{\hat a}\inv f\{y\}$. For the other direction, take $q=fp$ for any $p\in x \steps{\hat a} \inv f\{y\}$.

  Lastly, we show the satisfaction of \eqref{eq:2}. Let $x\in X,y\in Y$. Then we derive
        \begin{align*}
          P(x,\tau^*\hat a,\inv f \{y\})=&\ \mu_P((x \steps{\tau^*\hat a} \inv f\{y\})^\star) \quad (\text{Proposition~\ref{prop:separated}.\ref{prop:separated:future} and $\future U=\bigcup_{p\in U}\future p$})\\
            =&\ \mu_P\bigg(\Big(\bigcup_{p\in x \steps{\hat a} \inv f \{y\}} \future{p}\Big)^\star\bigg)\\
            =&\ \tilde\mu_P\bigg(\bigcup_{p\in x \steps{\hat a} \inv f \{y\}} \future{p}\bigg)\\
            =&\ \alpha(x)(\bigcup_{p\in x \steps{\hat a}\inv f \{y\}} \future{\qset{[p]}}).\qedhere
        \end{align*}
\end{proof}
\end{document}